\documentclass{article}
\usepackage{amsmath,amssymb}
\usepackage{longtable}
\usepackage{graphv1}
\usepackage{amsthm}

\newtheorem{theorem}{Theorem}
\newtheorem{lemma}{Lemma}
\newtheorem{definition}{Definition}
\newtheorem{remark}{Remark}
\newtheorem{corollary}{Corollary}
\newtheorem{proposition}{Proposition}
\newcounter{claimctr}[lemma]
\newenvironment{claim}{\medskip\par\noindent\refstepcounter{claimctr}\textbf{Claim \theclaimctr.} }{\smallskip\par}

\newtheorem{conjecture}{Conjecture}
\newcommand{\dash}[2]{\raisebox{.75mm}{\rule{#1mm}{#2mm}}}
\newcommand{\miss}{\dash{.6}{0}\dash{.75}{.2}\dash{.65}{0}\dash{.75}{.2}\dash{.65}{0}\dash{.75}{.2}\dash{.6}{0}}
\newcommand{\com}[1]{}
\newcounter{propctr}

\newenvironment{myenumerate}
{\vspace{-0.05in}
\begin{enumerate}
\itemsep -0.04in
\renewcommand{\theenumi}{\textup{($\alph{enumi}$)}}
\renewcommand{\labelenumi}{\theenumi}
}
{\end{enumerate}}
\newenvironment{proofobs}{%
  \vspace{-0.05in}
  \list{}{%
    \leftmargin0.5cm   
    \rightmargin\leftmargin
  }
  \item\relax
}
{\endlist}

\begin{document}
\title{Extending some results on the second neighborhood conjecture}
\author{Suresh Dara\footnote{Department of Mathematics, VIT Bhopal University. email: \texttt{suresh.dara@vitbhopal.ac.in}}\and Mathew C. Francis\footnote{Indian Statistical Institute, Chennai. email: \texttt{\{mathew,dalujacob\}@isichennai.res.in}}\and Dalu Jacob\footnotemark[2] \and N. Narayanan\footnote{Department of Mathematics, Indian Institute of Technology Madras. email: \texttt{naru@iitm.ac.in}}}
\date{}
\maketitle

\begin{abstract}
If in a directed graph, $v$ is an out-neighbor of $u$ and $w$ is an out-neighbor of $v$ but not of $u$, then $w$ is said to be a \emph{second out-neighbor} of $u$. A vertex in a directed graph is said to have a \emph{large second neighborhood} if it has at least as many second out-neighbors as out-neighbors. The Second Neighborhood Conjecture, first stated by Seymour, asserts that there is a vertex having a large second neighborhood in every oriented graph (a directed graph without loops or digons). It is straightforward to see that the conjecture is true for any oriented graph whose underlying undirected graph is bipartite. We extend this to show that the conjecture holds for oriented graphs whose vertex set can be partitioned into an independent set and a 2-degenerate graph. Fisher proved the conjecture for tournaments and later Havet and Thomass\'e provided a different proof for the same using median orders of tournaments. Havet and Thomass\'e in fact showed the stronger statement that if a tournament contains no sink, then it contains at least two vertices with large second neighborhoods. Using their techniques, Fidler and Yuster showed that the conjecture remains true for tournaments from which either a matching or a star has been removed. We extend this result to show that the conjecture holds even for tournaments from which both a matching and a star have been removed. This implies that a tournament from which a matching has been removed contains either a sink or two vertices with large second neighborhoods.
\end{abstract}
\section{Introduction}
Throughout this article, all graphs are finite and simple. Let $D =(V,E)$ be a digraph with vertex set $V(D)$ and arc set $E(D)$. As usual, $N_D^+(v)$ (resp. $N_D^-(v)$) denotes the out-neighborhood (resp. in-neighborhood) of a vertex $v\in V(D)$. Let $N_D^{++}(v)$ denote the second out-neighborhood of $v$, which is the set of vertices whose distance from $v$ is exactly 2, i.e. $N_D^{++}(v)=\{u\in V(D)\setminus (\{v\}\cup N^+_D(v)) \colon N^-_D(u)\cap N^+_D(v)\neq\emptyset\}$. The \emph{out-degree} of a vertex $v$ is defined to be $|N^+_D(v)|$. The minimum out-degree of $D$ is the minimum value among the out-degrees of all vertices of $D$. We omit the subscript if the digraph under consideration is clear from the context. 

A vertex $v$ in a digraph $D$ is said to have a \emph{large second neighborhood} if $|N^{++}(v)|\geq |N^+(v)|$. Oriented graphs are digraphs without loops or digons: i.e. they can be obtained by orienting the edges of a simple undirected graph. Paul Seymour conjectured the following in 1990 (see~\cite{dean1995squaring}):
\begin{conjecture}[The Second Neighborhood Conjecture]\label{ssnc}
Every oriented graph contains a vertex with a large second neighborhood.
\end{conjecture}
The above conjecture, if true, implies a special case of another open question concerning digraphs called the Caccetta-H\"aggkvist Conjecture~\cite{caccetta1978minimal}.
Note that a \emph{sink}---a vertex with out-degree zero---trivially has a large second neighborhood and therefore the conjecture is true for any oriented graph that contains a sink. It is easy to see that there exist oriented graphs in which the only vertex with a large second neighborhood is a sink (for example, any acyclic tournament).

Conjecture~\ref{ssnc} for the special case of tournaments, known as Dean's Conjecture~\cite{dean1995squaring}, was solved by Fisher~\cite{fisher1996squaring} in 1996 using some basic linear algebraic and probabilistic arguments. Later in 2000, Havet and Thomass\'e~\cite{havet2000median} gave a short combinatorial proof of Dean's Conjecture using ``median orders'' of tournaments. They could in fact prove something stronger: in a tournament without a sink, there exist two vertices with large second neighborhoods. Using the approach of Havet and Thomass\'e, Fidler and Yuster~\cite{fidler2007remarks} in 2007 proved that the Second Neighborhood Conjecture is true for oriented graphs that can be obtained from tournaments by removing edges in some specific ways. In particular, they showed that a tournament missing a matching (an oriented graph whose missing edges form a matching), a tournament missing a star and a tournament missing a complete graph all satisfy the conjecture. As these results hold even if the missing matching (or star, or complete graph) is empty, they extend the proof of Dean's Conjecture by Havet and Thomass\'e. Using techniques from this paper, Ghazal~\cite{ghazal2012seymour} proved that the Second Neighborhood Conjecture is true for tournaments missing a ``generalized star''---a $(P_4,C_4,2K_2)$-free graph---thereby extending the theorems of Fidler and Yuster for tournaments missing a star and tournaments missing a complete graph. It has to be noted that among these results that all use the median order approach, the case of the tournament missing a matching is by far the most difficult one, requiring a complicated proof. In this paper, we introduce new ideas to refine and extend this proof, allowing us to prove the conjecture for a superclass of tournaments missing a matching: we show that oriented graphs whose missing edges can be partitioned into a (possibly empty) matching and a (possibly empty) star also satisfy the Second Neighborhood Conjecture. In fact, we prove the stronger statement that in such a graph that does not contain a sink, there exists a vertex that has a large second neighborhood and is not the center of the missing star.

Ghazal~\cite{ghazal2015remark} attempts to generalize the theorem of Havet and Thomass\'e by trying to prove that there exist two vertices with large second neighborhoods in every tournament missing a matching that does not contain a sink. He shows that if a tournament missing a matching satisfies certain additional technical conditions, then such a result can be obtained. Our result mentioned above directly yields a proof that shows that every tournament missing a matching that does not contain a sink has at least two vertices with large second neighborhoods.

Other researchers have tried to attack special cases of the Second Neighborhood Conjecture without using the median order approach.
Llad\'{o}~\cite{llado} proved the conjecture in regular oriented graphs with high connectivity. Kaneko and Locke~\cite{kaneko2001minimum} verified the conjecture for oriented graphs with minimum out-degree at most 6. We state their result below as we use it later.

\begin{theorem}[\cite{kaneko2001minimum}]\label{mindegree}
Every oriented graph with minimum out-degree less than 7 has a vertex with a large second neighborhood.
\end{theorem}

It is easy to verify that in any oriented graph, a minimum out-degree vertex whose out-neighborhood is an independent set is a vertex with a large second neighborhood. Therefore, the conjecture is true for bipartite graphs (in fact, it is true if the underlying undirected graph is triangle-free). It appears difficult to prove the conjecture even for oriented graphs whose underlying undirected graph is 3-colourable.

For $S\subseteq V(D)$, we denote by $D[S]$ the graph induced in $D$ by the vertices in $S$ and by $D-S$ the graph $D[V(D)-S]$. An undirected graph $G$ is said to be \emph{2-degenerate} if every subgraph of $G$ has a vertex of degree at most two. We say that an oriented graph is 2-degenerate if its underlying undirected graph is 2-degenerate. We show that the conjecture is true for every oriented graph whose vertices can be partitioned into two sets such that one is an independent set and the other induces a 2-degenerate graph.
\medskip

\subsection*{Outline of the paper}
In Section~\ref{2-deg}, we show that the conjecture is true for any oriented graph $D$ such that $V(D)$ is the disjoint union of two sets $A$ and $B$ where $D[A]$ is 2-degenerate and $D[B]$ is an independent set. The proof relies on some counting arguments. In Section~\ref{matching}, we prove that any oriented graph whose missing edges can be partitioned into a (possibly empty) matching and a (possibly empty) star satisfies the Second Neighborhood Conjecture. This implies that a tournament missing a matching has at least two vertices with large second neighborhoods unless it contains a sink.
In conclusion, we ask whether it is true that if there is exactly one vertex with a large second neighborhood in an oriented graph, then it is a sink. We note that such a result would imply the Second Neighborhood Conjecture.

\section{Graphs that are almost bipartite}\label{2-deg}
In this section, we prove that the Second Neighborhood Conjecture is true for the class of oriented graphs whose vertex set has a partition $(A,B)$ such that $B$ is an independent set and the subgraph induced by $A$ is 2-degenerate. Note that every subgraph of a 2-degenerate graph is also 2-degenerate. 
\begin{proposition}\label{2-d}
Let $H=(V,E)$ be an oriented graph on $n$ vertices which is 2-degenerate. Then,
\begin{myenumerate}
\item\label{edges} If $n\geq 2$ then $|E(H)|\leq 2n-3$.
\item\label{outdegree} $H$ has at least one vertex with out-degree at most 1.
\end{myenumerate}
\end{proposition}
\begin{proof}
\ref{edges} We prove this by induction on $|V(H)|=n$. It is trivially true in the base case where $n=2$. Assume that the statement is true for all 2-degenerate graphs with less than $n$ vertices. As $H$ is 2-degenerate, it has a vertex of degree at most 2, say $x$. Now the subgraph $H-\{x\}$ of $H$ is itself 2-degenerate and has only $n-1$ vertices. Therefore, by the induction hypothesis, $|E(H-\{x\})|\leq 2(n-1)-3$. As $x$ has at most 2 edges incident to it, we have $|E(H)|\leq |E(H-\{x\})|+2\leq 2n-3$. 
		
\ref{outdegree} If $n=1$, then the statement is trivially true. So we can assume that $n\geq 2$. Note that $|E(H)|=\sum_{u\in V(H)} |N^+(u)|$. Therefore, if $|N^+(u)|\geq 2$ for every vertex $u\in V(H)$, then we would get $|E(H)|\geq 2n$, contradicting \ref{edges}.
\end{proof}
	
For the remainder of this section, we denote by $D=(V,E)$ an oriented graph whose vertex set has a partition $(A,B)$ such that $B$ is an independent set of $D$ and $D[A]$ is 2-degenerate.

Let $d$ be the minimum out-degree of $D$.

\begin{lemma}\label{v_in_B}
If there is a vertex in $B$ with out-degree $d$ in $D$, then $D$ has a vertex with large second neighborhood.
\end{lemma}
	\begin{proof}
		Suppose not. Let $v\in B$ be a vertex such that $|N^+(v)|=d$. If $d=0$, then $v$ is trivially a vertex with large second neighbourhood. Further, if $d=1$, then it can be seen that the single out-neighbour of $v$ has an out-neighbour that is different from $v$, and therefore $v$ is again a vertex with large second neighbourhood. Thus we can assume that $d\geq 2$. As $v\in B$ and $B$ is an independent set, we have $N^+(v)\subseteq A$. Let $N^{++}(v)=X\cup Y$, where $X\subseteq A$, $Y \subseteq B$. Also, let $|X|=x$ and $|Y|=y$. As $v$ does not have a large second neighborhood and $|N^+(v)|=d$, we have $x+y\leq d-1$. Consider the subgraph $H=D[N^+(v)\cup X\cup Y]$. As $N^+(v)\cup X\subseteq A$, $|N^+(v)\cup X|=d+x\geq 2$, and $D[A]$ is 2-degenerate, by Proposition~\ref{2-d}\ref{edges}, the maximum number of edges
        in $D[N^+(v)\cup X]$ is $2(d+x)-3$. Together with the at most $dy$ edges between $N^+(v)$ and $Y$, we get that the number of edges in $H$ that have at least one end point in $N^+(v)$ is at most $2(d+x)+dy-3$. i.e., $|\{(p,q)\in E(H)\colon \{p,q\}\cap N^+(v)\neq\emptyset\}|\leq 2(d+x)+dy-3$. Also since each vertex $u\in N^+(v)$ has out-degree at least $d$, we have that $|\{(p,q)\in E(H)\colon p\in N^+(v)\}|\geq d^2$. Therefore we can conclude that,
		\begin{equation}\label{eqn_case1}
		2d+2x+dy-3 \geq d^2
		\end{equation}
		Suppose that $y\leq d-2$. Then we have,
		$$\begin{array}{rclr}
		2d+2x+dy-3 &=& 2d+2(x+y)+(d-2)y-3&\text{(adding and subtracting }2y)\\&\leq& 2d+2(d-1)+(d-2)^2-3 &\text{(since } x+y\leq d-1\text{ and } y\leq d-2)\\ &=& d^2-1 < d^2& 
		\end{array}$$
		which is a contradiction to~\eqref{eqn_case1}. Therefore, $y\geq d-1$. Since $x+y\leq d-1$, this implies that $x=0$ and $y=d-1$. As $N^+(v)\subseteq A$, we know that $D[N^+(v)]$ is 2-degenerate. Then by Proposition~\ref{2-d}\ref{outdegree}, there exists a vertex $w\in N^+(v)$ whose out-degree in $D[N^+(v)]$ is at most 1. In fact, the out-degree of $w$ in $D[N^+(v)]$ is exactly 1, as otherwise $N^+(w)\subseteq Y$, implying that $y\geq |N^+(w)|\geq d$, which contradicts the fact that $y=d-1$. Let $w'$ be the unique out-neighbor of $w$ in $D[N^+(v)]$. Note that since $w\in N^+(v)$, we have $N^+(w)\subseteq N^+(v)\cup N^{++}(v)$, or in other words, $N^+(w)\subseteq N^+(v)\cup X\cup Y$. Then the fact that $N^+(w)\cap N^+(v)=\{w'\}$ and $x=0$ implies that $N^+(w)\subseteq\{w'\}\cup Y$. Since $y=d-1$, this further implies that $|N^+(w)|=d$; in particular, $N^+(w)=Y\cup\{w'\}$. Again, as with $w$, it can be seen that $N^+(w')\subseteq N^+(v)\cup X\cup Y$. As $x=0$ and $w'$ has at most $d-1$ out-neighbors in $N^+(v)$, it is clear that $w'$ should have at least one out-neighbor in $Y$, say $z$. Then $z\in N^+(w)\cap N^+(w')$. As $Y$ is an independent set, we have $N^+(z)\subseteq A\setminus\{w,w'\}$, implying that $N^+(z)$ is disjoint from $N^+(w)=Y\cup\{w'\}$. This means that $N^+(z)\subseteq N^{++}(w)$, which gives $|N^{++}(w)|\geq |N^+(z)|\geq d=|N^+(w)|$. Hence $w$ has a large second neighborhood in $D$, which is a contradiction. 
	\end{proof}
	\begin{lemma}\label{v_in_A}
		If every vertex in $B$ has out-degree at least $d+1$ in $D$, then $D$ has a vertex with large second neighborhood.
	\end{lemma}
	\begin{proof}
		Suppose not. Note that from Theorem~\ref{mindegree}, we have $d>6$. Clearly, there is a vertex $v\in A$ such that $|N^+(v)|=d$. Let $N^+(v)=X\cup Y$, where $X\subseteq A$ and $Y\subseteq B$, and $N^{++}(v)= X'\cup Y'$, where $X'\subseteq A$ and $Y'\subseteq B$. Also, let $|X|=x, |Y|=y, |X'|=x'$ and $|Y'|=y'$. Note that $x+y=d$, and since $v$ does not have a large second neighborhood, $x'+y'\leq d-1$. Since each vertex of $Y$ has at least $d+1$ out-neighbors, all of which lie in $X\cup X'$, we further have $x+x'\geq d+1$. 
		\begin{claim}\label{card(x)}
			$x\geq 3$.
		\end{claim}
		\begin{proofobs}
			Assume to the contrary that $x\leq 2$. Then since $x'\leq d-1$ and $x+x'\geq d+1$, it should be the case that $x=2$, $x'=d-1$, $y'=0$ and $x+x'=d+1$. This implies that $N^+(u)=X\cup X'$ for all $u\in Y$. Then, neither vertex in $X$ can have an out-neighbor in $Y$. Now if $w\in X$ is a vertex that has no out-neighbor in $X$ (clearly, such a vertex exists as $x=2$), the fact that $y'=0$ implies that $N^+(w)\subseteq X'$. But $x'=d-1$, implying that $|N^+(w)|<d$, which is a contradiction. This proves the claim. \qed
		\end{proofobs}
		
		Now, consider the subgraph $H=D[X\cup Y\cup X'\cup Y']$. As $X\cup X'\subseteq A$, $x+x'\geq d+1>7$ and $D[A]$ is 2-degenerate, by Proposition~\ref{2-d}\ref{edges}, the maximum number of edges in $D[X\cup X']$ is $2(x+x')-3$. Together with the at most $xy$ edges between $X$ and $Y$, the at most $xy'$ edges between $X$ and $Y'$ and the at most $yx'$ edges between $Y$ and $X'$, we get that the number of edges in $H$ with at least one end point in $N^+(v)= X\cup Y$ is at most $2(x+x')-3+xy+xy'+x'y$, i.e., $|\{(p,q)\in E(H)\colon \{p,q\}\cap N^+(v)\neq\emptyset\}|\leq 2(x+x')-3+xy+xy'+x'y$. There are at least $d$ edges going out from each vertex of $X$ and at least $d+1$ edges going out from each vertex of $Y$. Therefore, $|\{(p,q)\in E(H)\colon p\in X\}|\geq dx$ and $|\{(p,q)\in E(H)\colon p\in Y\}|\geq (d+1)y$. Altogether, we have $|\{(p,q)\in E(H)\colon p\in N^+(v)\}|\geq dx+(d+1)y=d^2+y$ (as $x+y=d$). Hence we can conclude that,
		\begin{equation}\label{eqn_case2}
		2(x+x')-3+xy+xy'+x'y \geq d^2+y 
		\end{equation}
		
		\begin{claim}\label{x,x'}
			At most one of $x$ and $x'$ can be greater than or equal to $\frac{d}{2}+1$.
		\end{claim}
		\begin{proofobs}
		Suppose for the sake of contradiction that $x= \frac{d}{2}+r$ and $x'=\frac{d}{2}+s$, where $r,s\geq 1$. As $x+y=d$ and $x'+y'\leq d-1$ we have $y=\frac{d}{2}-r$ and $y'\leq\frac{d}{2}-s-1$.
			By substituting these in the LHS of the equation~(\ref{eqn_case2}) we have, 
			\begin{eqnarray*}
				2(x+x')-3+xy+xy'+x'y&\leq& 2(d+r+s)-3+\left(\frac{d}{2}+r\right)\left(\frac{d}{2}-r\right) +\left(\frac{d}{2}+r\right)\left(\frac{d}{2}-s-1\right)\\
				&&\hspace{2in}+\left(\frac{d}{2}+s\right)\left(\frac{d}{2}-r\right)\\
				&\leq& \frac{3d^2}{4}+2d+r-3-r^2-\frac{d}{2}\quad(\text{as } r \geq 1 \text{ we have } rs\geq s)
			\end{eqnarray*}
			Combining the last inequality with~(\ref{eqn_case2}), we get
			\begin{eqnarray*}
				\frac{3d^2}{4}+2d+r-3-r^2-\frac{d}{2}&\geq&d^2+y\\
				6d+4r&\geq&d^2+4y+12+4r^2\\
				d^2+4r&>&d^2+4y+12+4r^2\quad(\text{as }d>6)\\
				4r&>&4y+12+4r^2
			\end{eqnarray*}
			This is a contradiction as $r\geq 1$ and $y\geq 0$. This proves the claim. \qed
		\end{proofobs}
		\medskip
		
		Now, consider the LHS of~(\ref{eqn_case2}).
		\begin{eqnarray}
		2(x+x')-3+xy+xy'+x'y &=& 2x+2x'-3+xy+y'(x+y)+x'(x+y)-xx'-yy'\nonumber\\
		&&\hspace{2in}(\text{adding and subtracting }xx'+yy')\nonumber\\
		&=& 2x+2x'-3+xy+d(x'+y')-xx'-yy' \quad(\text{as }x+y=d)\nonumber\\
		&\leq& 2x+2x'-3+xy+d(d-1)-xx'-yy' \quad(\text{as }x'+y'\leq d-1) \label{eqn_subcase}
		\end{eqnarray}
		Now, suppose that $x'\geq y+2$. Then~\eqref{eqn_subcase} implies,
		\begin{eqnarray*}
			2(x+x')-3+xy+xy'+x'y &\leq& 2x+2x'-3+xy+d(d-1)-x(y+2)-yy'\\
			&=& d^2+2x'-3-d-yy'
		\end{eqnarray*}
		Combining this inequality with~(\ref{eqn_case2}), we have
		\begin{eqnarray}
		d^2+2x'-3-d-yy'&\geq& d^2+y\nonumber\\
		2x'-3-d-yy'&\geq& y\label{y'=0}
		\end{eqnarray}
		Therefore we get,
		\begin{eqnarray*}
			2x'-3-d-yy'+2x&\geq&y+2x\\
			2(x+x')-3-d-yy'&\geq&2d-y\quad(\text{as }x+y=d)
		\end{eqnarray*}
		As $\max\{x,x'\}=d$ and by Claim~\ref{x,x'}, $\min\{x,x'\}\leq\frac{d}{2}+1$, we have $x+x'\leq \frac{3d}{2}+1$. Combining this with the above inequality, we have\\
		\begin{eqnarray*}
			3d-1-d-yy'&\geq&2d-y\\
			y&\geq&yy'+1
		\end{eqnarray*}
		This implies that $y'=0$. Then~(\ref{y'=0}) becomes
		\begin{eqnarray*}
			2x'-3-d&\geq&y\\
			2x'&\geq&x+2y+3\quad(\text{as }d=x+y)\\
			x'&\geq&y+3 \quad(\text{as }x\geq 3\text{ by Claim~\ref{card(x)}})
		\end{eqnarray*}
		Substituting this together with $y'=0$ in the RHS of~\eqref{eqn_subcase} we get,
		\begin{eqnarray*}
			2(x+x')-3+xy+xy'+x'y &\leq& 2x+2x'-3+xy+d(d-1)-x(y+3)\\
			&=& d^2+2x'-3-d-x
		\end{eqnarray*}
		Combining this with~(\ref{eqn_case2}), we have
		\begin{eqnarray*}
			d^2+2x'-3-d-x&\geq&d^2+y\\
			x'&\geq&d+\frac{3}{2}\quad(\text{as }x+y=d)
		\end{eqnarray*}
		which contradicts the fact that $x'+y'\leq d-1$. Therefore, we can assume that $x'\leq y+1$.
		In fact, $x'=y+1$, as otherwise, $x+x'<x+y+1=d+1$, which is a contradiction to our earlier observation that $x+x'\geq d+1$. Now, substituting $x'=y+1$ in the RHS of~\eqref{eqn_subcase}, we get
		\begin{eqnarray*}
			2(x+x')-3+xy+xy'+x'y &\leq& 2x+2(y+1)-3+xy+d^2-d-x(y+1)-yy'\\
			&=& d^2+d-1-x-yy'\quad(\text{as }x+y=d)
		\end{eqnarray*}
		Now, combining this with~(\ref{eqn_case2}), we have
		\begin{eqnarray*}
			d^2+d-1-x-yy'&\geq&d^2+y\\
			yy'+1&\leq&0 \quad(\text{as }x+y=d)
		\end{eqnarray*}
		which is a contradiction. This proves the lemma.
	\end{proof}
	\begin{theorem}
		Let $D=(V,E)$ be an oriented graph whose vertex set $V(D)$ has a partition $(A,B)$, such that $B$ is an independent set and $D[A]$ is 2-degenerate. Then $D$ has a vertex with a large second neighborhood.
	\end{theorem}
	\begin{proof}
		The proof is immediate from Lemma~\ref{v_in_B} and Lemma~\ref{v_in_A}.
	\end{proof}
	
\section{Graphs that are almost tournaments}\label{matching}
In this section, our main aim will be to show that Conjecture~\ref{ssnc} is true for tournaments whose missing edges can be partitioned into a matching and a star. In Section~\ref{sec:median}, we review ``median orders'' of tournaments and their properties. We then study tournaments missing a matching in Section~\ref{sec:main}, wherein we introduce the notions and structural results that we need to prove our main result. 
We prove our main result about tournaments missing a matching and a star in Section~\ref{sec:nearmatching}.

\subsection{Median orders of tournaments}\label{sec:median}

Given an ordering of the vertices of a tournament, an arc of the tournament is said to be a ``forward arc'' if the starting vertex of the arc occurs earlier than its ending vertex in the ordering. A \emph{median order} of a tournament is an ordering of its vertices with the most number of forward arcs. Formally, an ordering $(x_1,x_2,\ldots,x_n)$ of the vertices of a tournament $T$ that maximizes $|\{(x_i,x_j)\in E(T)\colon i<j\}|$ is said to be a median order of $T$. The \emph{feed vertex} of a median order $(x_1,x_2,\ldots,x_n)$ is the last vertex $x_n$ in that ordering of vertices. Havet and Thomass\'e~\cite{havet2000median} proved the following. 

\begin{theorem}[Havet-Thomass\'e]\label{thmht}
Let $T$ be a tournament and $d$ be the feed vertex of a median order of $T$. Then $|N^+(d)|\leq |N^{++}(d)|$, i.e. $d$ has a large second neighborhood.
\end{theorem}

The following properties of median orders of tournaments are not difficult to verify (see~\cite{havet2000median}).

\begin{proposition}\label{subtournament}
If $(x_1,x_2,\ldots,x_n)$ is a median order of a tournament $T$ and let $x_i$ and $x_j$ be such that $1\leq i\leq j\leq n$. Let $T'=T[\{x_i,x_{i+1},\ldots,x_j\}]$. Then:
\begin{myenumerate}
\item\label{suborder} $(x_i,x_{i+1},\ldots,x_j)$ is a median order of $T'$, and
\item\label{replace} if $(y_1,y_2,\ldots,y_{j-i+1})$ is a median order of $T'$, then $(x_1,x_2,\ldots,x_{i-1},y_1,y_2,\ldots,y_{j-i+1},x_{j+1},x_{j+2},\ldots,$ $x_n)$ is a median order of $T$.
\end{myenumerate} 
\end{proposition}

\begin{proposition}\label{feedback}
Let $(x_1,x_2,\ldots,x_n)$ be a median order of a tournament $T$ and let $x_i$ and $x_j$ be such that $1\leq i<j\leq n$. Then:
\begin{myenumerate}
\item\label{xif} $|N^+(x_i)\cap\{x_{i+1},\ldots,x_j\}|\geq |N^-(x_i)\cap\{x_{i+1},\ldots,x_j\}|$, and
\item\label{xjf} $|N^+(x_j)\cap\{x_i,\ldots,x_{j-1}\}|\leq |N^-(x_j)\cap\{x_i,\ldots,x_{j-1}\}|$.
\end{myenumerate}
\end{proposition}

\begin{proposition}\label{movexi}
Let $(x_1,x_2,\ldots,x_n)$ be a median order of a tournament $T$ and let $x_i$ and $x_j$ be such that $1\leq i<j\leq n$. Then:
\begin{myenumerate}
\item\label{xi} If $|N^+(x_i)\cap\{x_{i+1},\ldots,x_j\}|=|N^-(x_i)\cap\{x_{i+1},\ldots,x_j\}|$, then $(x_1,x_2,\ldots,x_{i-1},x_{i+1},x_{i+2},\ldots,x_j,x_i,$ $x_{j+1},x_{j+2},\ldots,x_n)$ is also a median order of $T$, and
\item\label{xj} if $|N^+(x_j)\cap\{x_i,\ldots,x_{j-1}\}|=|N^-(x_j)\cap\{x_i,\ldots,x_{j-1}\}|$, then $(x_1,x_2,\ldots,x_{i-1},x_j,x_i,x_{i+1},\ldots,x_{j-1},$ $x_{j+1},x_{j+2},\ldots,x_n)$ is also a median order of $T$.
\end{myenumerate} 
\end{proposition}
\begin{proposition}\label{reversing}
Let $L=(x_1,x_2,\ldots,x_n)$ be a median order of a tournament $T$ and let $(x_j,x_i)\in E(T)$, where $i<j$. Then $L$ is also a median order of the tournament $T'$ with $V(T')=V(T)$ and $E(T')=(E(T)\setminus\{(x_j,x_i)\})\cup\{(x_i,x_j)\}$.
\end{proposition}
\begin{proof}
If $L$ is not a median order of $T'$, then there exists an ordering $\hat{L}$ of $V(T')=V(T)$ such that $(T',\hat{L})$ has at least one more forward arc than $(T',L)$ and therefore at least two more forward arcs than $(T,L)$. But then $(T,\hat{L})$ has at least one more forward arc than $(T,L)$, contradicting the fact that $L$ is a median order of $T$. Therefore, $L$ is a median order of $T'$ as well.
\end{proof}

\subsubsection*{Modules}
Given an oriented graph $D$, a set $S\subseteq V(D)$ is said to be a \emph{module} in $D$, if for any two vertices $u,v\in S$, $N^+(u)\setminus S= N^+(v)\setminus S$ and $N^-(u)\setminus S= N^-(v)\setminus S$. Trivially, a single vertex is a module by itself and so is the set $V(D)$.

\begin{proposition}\label{modulesec}
Let $D$ be an oriented graph and $S$ a module in it.
\begin{myenumerate}
\item\label{removemod} For $u\in S$, let $D'=D-(S\setminus\{u\})$. Then, $N^{++}_{D'}(u)=N^{++}_D(u)\setminus S$.
\item\label{outsidemod} For $u,v\in S$, $N^{++}_D(u)\setminus S= N^{++}_D(v)\setminus S$.
\end{myenumerate}
\end{proposition}
\begin{proof}
Clearly, $N^{++}_{D'}(u)\subseteq N^{++}_D(u)\setminus S$.
Consider any vertex $x\in N^{++}_D(u)\setminus S$. Then $(u,x)\notin E(D)$ and there exists $w\in V(D)$ such that $(u,w),(w,x)\in E(D)$.
As we have $x\notin S$, $(w,x)\in E(D)$, $(u,x)\notin E(D)$, and $S$ is a module containing $u$, we have $w\notin S$. Then since $u,w,x\in V(D')$, we have that $(u,w),(w,x)\in E(D')$ and $(u,x)\notin E(D')$, implying that $x\in N^{++}_{D'}(u)$. Therefore, $N^{++}_D(u)\setminus S\subseteq N^{++}_{D'}(u)$, proving~\ref{removemod}.

Note that for proving~\ref{outsidemod}, we only need to prove that $N^{++}_D(u)\setminus S\subseteq N^{++}_D(v)\setminus S$, as $u$ and $v$ are symmetric. Consider any vertex $x\in N^{++}_D(u)\setminus S$. As noted above, $(u,x)\notin E(D)$ and there exists $w\in V(D)$ such that $(u,w),(w,x)\in E(D)$. Since $u$ and $v$ belong to the module $S$ in $D$, we have that $(v,w)\in E(D)$ and $(v,x)\notin E(D)$, implying that $x\in N^{++}_D(v)\setminus S$. Therefore, $N^{++}_D(u)\setminus S\subseteq N^{++}_D(v)\setminus S$.
\end{proof}
\subsubsection*{Good median orders}
We now define a special kind of median order of tournaments, along the lines of Ghazal~\cite{ghazal2015remark}.
Given a partition $\mathcal{I} = \{I_1, I_2,\ldots, I_r\}$ of $V(T)$ such that each $I_i$, $1\leq i\leq r$, is a module in $T$, we say that a median order of $T$ is a \emph{good median order with respect to $\mathcal{I}$} if for each $i\in\{1,2,\ldots,r\}$, the vertices of $I_i$ appear consecutively in it (note that this is slightly different from the ``good median orders'' defined by Ghazal~\cite{ghazal2015remark}). The following lemma is implicit in the work of Ghazal (the interested reader may refer to the Appendix for a proof).

\begin{lemma}\label{module}
Let $\mathcal{I} = \{I_1, I_2,\ldots, I_r\}$ be a partition of the vertex set of a tournament $T$ into modules and let $L$ be a median order of $T$. Then there is a good median order $L'$ of $T$ with respect to $\mathcal{I}$ such that $L$ and $L'$ have the same feed vertex.
\end{lemma}

The following fact was noted by Ghazal~\cite{ghazal2015remark}.
\begin{proposition}\label{modht}
Let $d$ be the feed vertex of a median order of a tournament $T$ and let $I$ be a module in $T$ containing $d$. Then for any vertex $v\in I$, $|N^+(v)\setminus I|\leq |N^{++}(v)\setminus I|$.
\end{proposition}
\begin{proof}
Let $\mathcal{I}=\{I\}\cup\{\{u\}\colon u\notin I\}$. It is easy to see that $\mathcal{I}$ is a partition of $V(T)$ into modules.
By Lemma~\ref{module}, there exists a good median order $L=(x_1,x_2,\ldots,x_n=d)$ of $T$ with respect to $\mathcal{I}$. Then, there exists $i\in\{1,2,\ldots,n\}$ such that $I=\{x_i,x_{i+1},\ldots,x_n\}$. By Proposition~\ref{subtournament}\ref{suborder}, $L'=(x_1,x_2,\ldots,x_i)$ is a median order of $T'=T-(I\setminus\{x_i\})$. By Theorem~\ref{thmht}, $|N^+_{T'}(x_i)|\leq |N^{++}_{T'}(x_i)|$. Consider any $v\in I$. As $I$ is a module containing $x_i$ and $v$, $N^+_T(v)\setminus I=N^+_T(x_i)\setminus I=N^+_{T'}(x_i)$. By Proposition~\ref{modulesec}, we also have that $N^{++}_{T'}(x_i)=N^{++}_T(x_i)\setminus I=N^{++}_T(v)\setminus I$. Combining the above observations, we get $|N^+_T(v)\setminus I|\leq |N^{++}_T(v)\setminus I|$.
\end{proof}

\subsubsection*{Sedimentation of a good median order}

Havet and Thomass\'e~\cite{havet2000median} defined the process of ``sedimentation'' of a median order, using which a median order can be transformed into another median order when certain conditions are satisfied.
Ghazal modified the notion of sedimentation of median orders to apply to good median orders. We slightly modify this so as to redefine sedimentation without referring to the ``good'' and ``bad'' vertices that appear in the work of Havet and Thomass\'e and that of Ghazal.

Suppose that $L = (x_1,x_2,\ldots,x_n)$ is a good median order of a tournament $T$ with respect to $\mathcal{I}$, where $\mathcal{I}$ is a partition of $V(T)$ into modules. Let $I$ be the set in $\mathcal{I}$ containing $x_n$ and $t=|I|$. Then $I=\{x_{n-t+1},x_{n-t+2},\ldots,x_n\}$. Recall that by Proposition~\ref{modht}, $|N^+(x_n)\setminus I|\leq |N^{++}(x_n)\setminus I|$.
Then the \emph{sedimentation of $L$ with respect to $\mathcal{I}$}, denoted by $Sed_{\mathcal{I}}(L)$, is an ordering of $V(T)$ that is defined in the following way. If $|N^+(x_n)\setminus I| < |N^{++}(x_n)\setminus I|$, then $Sed_{\mathcal{I}}(L) = L$. If $|N^+(x_n)\setminus I| = |N^{++}(x_n)\setminus I|$, then $Sed_{\mathcal{I}}(L)$ is defined as follows. Let $b_1,b_2,\ldots,b_k$ be the vertices in $N^-(x_n)\setminus N^{++}(x_n)$ which are outside $I$ and $v_1,v_2,\ldots,v_{n-t-k}$ the vertices in $N^+(x_n)\cup N^{++}(x_n)$ which are outside $I$, both enumerated in the order in which they appear in $L$ (note that in any tournament, $N^{++}(u)\subseteq N^-(u)$ for any vertex $u$ in it). Then $Sed_{\mathcal{I}}(L)$ is the order $(b_1,b_2,\ldots,b_k,x_{n-t+1},x_{n-t+2},\ldots,x_n,v_1,v_2,\ldots,v_{n-t-k})$.

Given below is the main theorem that we need for sedimentation of median orders. This is a slight modification of a result of Ghazal~\cite{ghazal2015remark} to apply to our version of sedimentation (we again want to avoid using the ``good vertices'' of Havet and Thomass\'e; the interested reader may refer to the Appendix for a proof).
\begin{theorem}\label{sedimentation}
Let $T$ be a tournament. If $\mathcal{I}$ is a partition of $V(T)$ into modules and $L$ is a good median order of $T$ with respect to $\mathcal{I}$, then $Sed_{\mathcal{I}}(L)$ is also a good median order of $T$ with respect to $\mathcal{I}$.
\end{theorem}

Following Ghazal and Havet and Thomass\'e, we inductively define $Sed^0_{\mathcal{I}}(L) = L$ and for integer $q\geq 1$, $Sed^q_{\mathcal{I}}(L)=Sed_{\mathcal{I}}(Sed^{q-1}_{\mathcal{I}}(L))$. A good median order $L$ of $T$ with respect to some $\mathcal{I}$ is said to be \emph{stable} if there exists integer $q\geq 0$ such that $Sed^{q+1}_{\mathcal{I}}(L) = Sed^q_{\mathcal{I}}(L)$. Otherwise, $L$ is \textit{periodic}.

\subsection{Tournaments missing a matching}\label{sec:main}
In this section, we outline the proof of Fidler and Yuster that shows that the Second Neighborhood Conjecture is true for tournaments missing a matching.
Throughout this section, we denote by $D$ an oriented graph that can be obtained from a tournament by removing a (possibly empty) matching.

For a vertex $u\in V(D)$, we say that the vertices in $N^+_D(u)\cup N^-_D(u)$ are the \emph{neighbors} of $u$ and that the vertices in $V(D)\setminus (N^+_D(u)\cup N^-_D(u))$ are the \emph{non-neighbors} of $u$. It is easy to see that every vertex in $D$ has at most one non-neighbor. If there is no edge between two distinct vertices $x$ and $y$ in $D$, i.e., $x$ is a non-neighbor of $y$ (and vice versa), then we say that $\{x,y\}$ is a \emph{missing edge} in $D$. We denote this missing edge as $x \miss y$ (or, equivalently $y \miss x$). For an arc $(x,y)\in E(D)$, we use the notation $x\rightarrow y$ (in other words, $y\in N^+_D(x)$). If $(x,y)\in E(D)$ is an arc with the additional property that $x\notin N_D^{++}(y)$, then we say that $(x,y)$ is a \emph{special arc}, and denote it as $x\twoheadrightarrow y$. Note that there can be no directed triangle in $D$ containing a special arc.

\begin{lemma}\label{lemcycle}
	Let $C = a_0\rightarrow a_1\twoheadrightarrow a_2\twoheadrightarrow a_3
	\twoheadrightarrow\cdots\twoheadrightarrow a_{k-1}\rightarrow a_0$ be a cycle in $D$. Then:
	\begin{myenumerate}
		\item\label{missone} $a_0$ has a non-neighbor in $C$, and
		\item\label{inout} if $a_0\miss a_i$, then for $j\in\{1,\ldots,i-1\}$, $a_0\rightarrow a_j$ and for $j\in\{i+1,\ldots,k-1\}$, $a_j\rightarrow a_0$.
	\end{myenumerate}
\end{lemma}
\begin{proof}
	Since $D$ is an oriented graph that has no directed triangle containing a special arc, we have that $k\geq 4$.
	
	\ref{missone} Assume to the contrary that $a_0$ has no non-neighbor in $C$, i.e., $\forall i\neq 0, a_0\rightarrow a_i$ or $a_i\rightarrow a_0$. For some $i\neq 0$, if $a_0\rightarrow a_i$, then $a_0\rightarrow a_{i+1}$, because otherwise, $a_0\rightarrow a_i\twoheadrightarrow a_{i+1}\rightarrow a_0$ forms a directed triangle containing a special arc. Now since $a_0\rightarrow a_1$, applying this observation repeatedly gives us $a_0\rightarrow a_2$, $a_0\rightarrow a_3$, \ldots, $a_0\rightarrow a_{k-1}$, which is a contradiction to the fact that $a_{k-1}\rightarrow a_0$.
	
	\ref{inout} Let $a_0\miss a_i$. As $a_i$ is the only non-neighbor of $a_0$ in $D$, for each $j\notin\{0,i\}$, we have either $a_0\rightarrow a_j$ or $a_j\rightarrow a_0$. Suppose that for some $j\in\{1,\ldots,i-1\}$, we have $a_j\rightarrow a_0$, then consider the cycle $C'=a_0\rightarrow a_1\twoheadrightarrow \cdots\twoheadrightarrow a_j\rightarrow a_0$. Then $a_0$ has no non-neighbor in $C'$, which is a contradiction to~\ref{missone}. Similarly, if there is some $j\in\{i+1,\ldots,k-1\}$ such that $a_0\rightarrow a_j$, then there is no non-neighbor of $a_0$ in the cycle $a_0\rightarrow a_j\twoheadrightarrow a_{j+1}\twoheadrightarrow\cdots\twoheadrightarrow a_{k-1}\rightarrow a_0$, again contradicting~\ref{missone}.
\end{proof}

\subsubsection*{Special cycles}

We call a cycle in $D$ a \emph{special cycle} if it consists only of special arcs. It is easy to see that any special cycle contains at least 4 vertices.
The following corollary is an immediate consequence of Lemma~\ref{lemcycle}.

\begin{corollary}\label{propcycle}
	Let $C = a_0\twoheadrightarrow a_1\twoheadrightarrow a_2\twoheadrightarrow a_3\twoheadrightarrow\cdots\twoheadrightarrow a_{k-1}\twoheadrightarrow a_0$ be a special cycle in $D$. Then:
	\begin{myenumerate}
		\item\label{missonespecial} Each vertex in $C$ has a non-neighbor in $C$, 
		\item\label{inoutspecial} if $a_i\miss a_j$, then $N^+_D(a_i)\cap V(C)=\{a_{i+1},a_{i+2},\ldots,a_{j-1}\}$ and $N^-_D(a_i)\cap V(C)=\{a_{j+1},a_{j+2},\ldots,a_{i-1}\}$, where subscripts are modulo $k$.
	\end{myenumerate}
\end{corollary}

\begin{lemma}\label{cycle_cross_mod}
	Let $C = a_0\twoheadrightarrow a_1\twoheadrightarrow a_2\twoheadrightarrow a_3\twoheadrightarrow\cdots\twoheadrightarrow a_{k-1}\twoheadrightarrow a_0$ be a special cycle in $D$. Then:
	\begin{myenumerate}
		\item\label{kiseven} $k$ is even,
		\item\label{cross} For each vertex $a_i\in V(C)$, $a_i\miss a_{i+\frac{k}{2}}$ (subscripts modulo $k$),
		\item\label{module_cycle} $V(C)$ forms a module in $D$.
	\end{myenumerate}
\end{lemma}
\begin{proof}
	Using Corollary~\ref{propcycle}\ref{missonespecial}, we have that every vertex of $C$ has exactly one non-neighbor in $C$. This proves~\ref{kiseven}.
	
	\ref{cross} Let $a_j$ be the non-neighbor of $a_i$ in $C$. Suppose that $j\neq i+\frac{k}{2}$ (modulo $k$). Then one of the sets $\{a_{i+1},a_{i+2},\ldots,a_{j-1}\},\{a_{j+1},a_{j+2},\ldots,a_{i-1}\}$ (subscripts modulo $k$) is larger than the other. We shall assume without loss of generality that $|\{a_{i+1},a_{i+2},\ldots,a_{j-1}\}|>|\{a_{j+1},a_{j+2},\ldots,a_{i-1}\}|$. This means that there exists $a_p,a_q\in\{a_{i+1},a_{i+2},\ldots,a_{j-1}\}$ such that $a_p\miss a_q$, where $a_p$ occurs before $a_q$ in the ordering $a_{i+1},a_{i+2},\ldots,a_{j-1}$. By Corollary~\ref{propcycle}\ref{inoutspecial}, we know that $a_i\rightarrow a_q$. Now consider the cycle $C'=a_q\twoheadrightarrow a_{q+1}\twoheadrightarrow\cdots\twoheadrightarrow a_{i-1}\twoheadrightarrow a_i\rightarrow a_q$ (subscripts modulo $k$). There is no non-neighbor of $a_q$ in $C'$ (as $a_p$ is the only non-neighbor of $a_q$), which contradicts Lemma~\ref{lemcycle}\ref{missone}.
	\ref{module_cycle} Since every vertex of $C$ has a non-neighbor in $C$, for any $x\in V(D)\setminus V(C)$, $x$ is a neighbor of every vertex in $V(C)=\{a_0,a_1,\ldots,a_{k-1}\}$. This implies that if $x\rightarrow a_i$ for any $i\in\{0,1,\ldots,k-1\}$, then we also have $x\rightarrow a_{i+1}$ as otherwise, $x\rightarrow a_i\twoheadrightarrow a_{i+1}\rightarrow x$ would be a directed triangle containing a special arc (subscripts modulo $k$). Therefore applying this observation repeatedly starting from $a_0$, we get $N^-_D(a_0)\setminus V(C)\subseteq N^-_D(a_1)\setminus V(C)\subseteq  N^-_D(a_2)\setminus V(C)\subseteq \cdots \subseteq N^-_D(a_{k-2})\setminus V(C)\subseteq N^-_D(a_{k-1})\setminus V(C)\subseteq N^-_D(a_0)\setminus V(C)$. Similarly, if $a_i\rightarrow x$ for any $i\in\{0,1,\ldots,k-1\}$, then we also have $a_{i-1}\rightarrow x$, as otherwise $x\rightarrow a_{i-1}\twoheadrightarrow a_i\rightarrow x$ would be a directed triangle containing a special arc (subscripts modulo $k$). Again applying this observation repeatedly starting from $a_0$, we get $N^+_D(a_0)\setminus V(C)\subseteq N^+_D(a_{k-1})\setminus V(C)\subseteq N^+_D(a_{k-2})\setminus V(C)\subseteq\cdots \subseteq N^+_D(a_2)\setminus V(C)\subseteq N^+_D(a_1)\setminus V(C)\subseteq N^+_D(a_0)\setminus V(C)$. This shows that for any two vertices $a_i,a_j\in V(C)$, we have $N^+_D(a_i)\setminus V(C)=N^+_D(a_j)\setminus V(C)$ and $N^-_D(a_i)\setminus V(C)=N^-_D(a_j)\setminus V(C)$, implying that $V(C)$ forms a module in $D$.
\end{proof}

\subsubsection*{The relation $R$ and the digraph $\Delta_D$}
Let $M$ be the set $\{(x,y)\in V(D)\times V(D)\colon x \miss y\}$. We define a relation $R$ on $M$ as follows. For distinct $(a,b),(c,d)\in M$, we say that $(a,b)R(c,d)$ if and only if there exists the four cycle $a\rightarrow c\twoheadrightarrow b\rightarrow d\twoheadrightarrow a$ in $D$ (refer Figure~\ref{fig:relation}). Note that $(a,b)R(c,d)$ if and only if $(b,a)R(d,c)$. Following Fidler and Yuster~\cite{fidler2007remarks}, we now define an auxiliary digraph $\Delta_D$ whose vertices are the missing edges of $D$. This graph has the vertex set $V(\Delta_D)=\{\{a,b\}\colon a,b\in V(D)$ and $a\miss b\}$ and arc set $E(\Delta_D)=\{(\{a,b\},\{c,d\})\colon (a,b)R(c,d)\}$. In other words, there is an arc from the vertex $\{a,b\}$ to the vertex $\{c,d\}$ in $\Delta_D$ if and only if either $(a,b)R(c,d)$ or $(a,b)R(d,c)$. Note that from the definition of $R$, we cannot have both $(a,b)R(c,d)$ and $(a,b)R(d,c)$.

\begin{figure}[h]
\newcommand{\myptr}{{\arrow{>};}}
\newcommand{\spec}{{\arrow{>>};}}
\renewcommand{\arrowplacement}{0.5}
\renewcommand{\vertexset}{(a,0,0),(b,1,0),(c,3,0),(d,4,0)}
\renewcommand{\edgeset}{(a,b,,,,dashed),(c,d,,,,dashed),(a,c,,,30,,myptr),(b,d,,,30,,myptr),(c,b,,,20,,spec),(d,a,,,40,,spec)}
\renewcommand{\defradius}{.1}
\begin{center}
\begin{tikzpicture}
\drawgraph
\node[left] at (\xy{a}) {$a$};
\node[right] at (\xy{b}) {$b$};
\node[left] at (\xy{c}) {$c$};
\node[right] at (\xy{d}) {$d$};
\end{tikzpicture}
\end{center}
\caption{Situation that leads to $(a,b)R(c,d)$.}\label{fig:relation}
\end{figure}
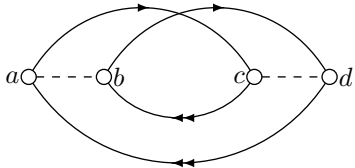

\begin{lemma}[Fidler-Yuster~\cite{fidler2007remarks}]\label{indegree}
For any vertex $e\in V(\Delta_D)$, $|N^+(e)|\leq1$ and $|N^-(e)|\leq1$.
\end{lemma}
\begin{proof}
	Let $e=\{a,b\}$. Suppose that it has two out-neighbors in $\Delta_D$, say $e_1 = \{c_1,d_1\}$, $e_2= \{c_2,d_2\}$. Recalling the definition of $\Delta_D$, we can assume without loss of generality that $(a,b)R(c_1,d_1)$ and $(a,b)R(c_2,d_2)$. That is, we have $a\rightarrow c_1\twoheadrightarrow b\rightarrow d_1\twoheadrightarrow a$ and $a\rightarrow c_2\twoheadrightarrow b\rightarrow d_2\twoheadrightarrow a$ in $D$. As $d_1$ is already a non-neighbor of $c_1$, we cannot have $c_1\miss d_2$. Now if $c_1\rightarrow d_2$ then we have the directed triangle $a\rightarrow c_1\rightarrow d_2 \twoheadrightarrow a$ containing a special arc, which is a contradiction. Similarly, if $d_2\rightarrow c_1$ then $b\rightarrow d_2\rightarrow c_1\twoheadrightarrow b$ is a directed triangle containing a special arc, which is again a contradiction. Thus, $|N^+(e)|\leq1$.
	
	Now suppose $e=\{a,b\}$ has two in-neighbors in $\Delta_D$, say $e_1 = \{c_1,d_1\}$, $e_2= \{c_2, d_2\}$. Again, we can assume without loss of generality that $(c_1,d_1)R(a,b)$ and $(c_2,d_2)R(a,b)$. Then we have $c_1\rightarrow a\twoheadrightarrow d_1\rightarrow b\twoheadrightarrow c_1$ and $c_2\rightarrow a\twoheadrightarrow d_2\rightarrow b\twoheadrightarrow c_2$ in $D$. As before, we cannot have $c_1\miss d_2$. If $c_1\rightarrow d_2$ then we have the directed triangle $c_1\rightarrow d_2\rightarrow b \twoheadrightarrow c_1$ containing a special arc and if $d_2\rightarrow c_1$, we have another directed triangle $d_2\rightarrow c_1\rightarrow a \twoheadrightarrow d_2$ containing a special arc. Since we have a contradiction in both cases, we conclude that $|N^-(e)|\leq1$.
\end{proof}

Therefore, $\Delta_D$ is a disjoint union of directed paths and directed cycles. Let $\mathcal{P}$ denote the collection of these directed paths and $\mathcal{C}$ denote the collection of these directed cycles.

For a cycle $Q\in\mathcal{C}$, we let $\Gamma(Q)=\bigcup_{\{u,v\}\in V(Q)} \{u,v\}$. That is, if $Q=\{a_1,b_1\}\{a_2,b_2\}\cdots\{a_t,b_t\}\{a_1,b_1\}$, then $\Gamma(Q)=\{a_1,b_1,a_2,b_2,\ldots,a_t,b_t\}$. Since $D$ is a tournament missing a matching, it is clear that the sets in $\{\Gamma(Q):Q\in\mathcal{C}\}$ are all pairwise disjoint.

\begin{lemma}\label{special_cycle}
	Let $Q\in\mathcal{C}$. Then there exists a special cycle $C$ in $D$ such that $V(C)=\Gamma(Q)$. 
\end{lemma}
\begin{proof}
	Let $Q=\{a_1,b_1\}\{a_2,b_2\}\cdots\{a_k,b_k\}\{a_1,b_1\}$. Note that $a_i\miss b_i$, for $1\leq i\leq k$. We shall assume that $k$ is even as the case when $k$ is odd is similar. Also, we can assume without loss of generality that for every $i\in\{1,2,\ldots, k-1\}$, $(a_i,b_i)R(a_{i+1},b_{i+1})$ (since we can always exchange the labels of $a_i$ and $b_i$, if required, so that this condition is satisfied). Then by the definition of $R$, we have $a_i\rightarrow a_{i+1}\twoheadrightarrow b_i\rightarrow b_{i+1}\twoheadrightarrow a_i$ for each $i\in\{1,2,\ldots,k-1\}$. Now if $(a_k,b_k)R (a_1,b_1)$ then we have $a_k\rightarrow a_1\twoheadrightarrow b_k\rightarrow b_1\twoheadrightarrow a_k$ (so $k>2$, implying that $k\geq 4$). This together with the previous observation implies that $C=a_1\twoheadrightarrow b_k\twoheadrightarrow a_{k-1}\twoheadrightarrow b_{k-2}\twoheadrightarrow a_{k-3}\twoheadrightarrow\cdots\twoheadrightarrow b_2\twoheadrightarrow a_1$ (as $k$ is even) is a special cycle in $D$, which contains only those $a_i$'s where $i$ is odd and those $b_i$'s where $i$ is even. This contradicts Corollary~\ref{propcycle}\ref{missone}, as for any odd $i$, the only non-neighbor $b_i$ of $a_i$ is not contained in $C$. Therefore, we have $(a_k,b_k)R (b_1,a_1)$. Then, $a_k\rightarrow b_1\twoheadrightarrow b_k\rightarrow a_1\twoheadrightarrow a_k$, which when combined with the previous observations gives us that $C= a_1\twoheadrightarrow a_k\twoheadrightarrow b_{k-1}\twoheadrightarrow a_{k-2}\twoheadrightarrow b_{k-3}\twoheadrightarrow\cdots \twoheadrightarrow a_2\twoheadrightarrow b_1\twoheadrightarrow b_k\twoheadrightarrow a_{k-1}\twoheadrightarrow b_{k-2}\twoheadrightarrow a_{k-3}\twoheadrightarrow\cdots\twoheadrightarrow b_2\twoheadrightarrow a_1$ is a special cycle in $D$ with $V(C)=\Gamma(Q)$.
\end{proof}

\begin{corollary}\label{gamma}
	Let $Q\in\mathcal{C}$. Then:
	\begin{myenumerate}
		\item\label{missingingamma} For each $u\in\Gamma(Q)$, there exists $v\in\Gamma(Q)$ such that $u\miss v$, and
		\item\label{gammamodule} $\Gamma(Q)$ forms a module in $D$.
	\end{myenumerate}
\end{corollary}
\begin{proof}
	The proof of~\ref{missingingamma} is immediate from Lemma~\ref{special_cycle} and Corollary~\ref{propcycle}\ref{missonespecial}.
	Similarly,~\ref{gammamodule} is a direct consequence of Lemma~\ref{special_cycle} and Lemma~\ref{cycle_cross_mod}\ref{module_cycle}.
\end{proof}
\begin{lemma}\label{snc_cycle}
	Let $Q\in \mathcal{C}$. Then for each $u\in\Gamma(Q)$, we have $|N^+_D(u)\cap\Gamma(Q)|=|N^{++}_D(u)\cap\Gamma(Q)|$.
\end{lemma}
\begin{proof}
	As $Q\in \mathcal{C}$, by Lemma~\ref{special_cycle} there exists a special cycle $C$ in $D$ such that $V(C)=\Gamma(Q)$. Let this cycle be $C= a_0\twoheadrightarrow a_1\twoheadrightarrow a_2\twoheadrightarrow \cdots\twoheadrightarrow a_{2l-1}\twoheadrightarrow a_0$ (note that by Lemma~\ref{cycle_cross_mod}\ref{kiseven}, $C$ has even length; also note that $l\geq 2$). Consider a vertex $a_i\in V(C)$. By Lemma~\ref{cycle_cross_mod}\ref{cross}, we have $a_i\miss a_{i+l}$ and by Corollary~\ref{propcycle}\ref{inoutspecial}, $N^+_D(a_i)\cap V(C)=\{a_{i+1},a_{i+2},\ldots,a_{i+l-1}\}$ (subscripts modulo $2l$). Recalling that $V(C)=\Gamma(Q)$, we now get $|N^+_D(a_i)\cap\Gamma(Q)|= l-1$. Now, consider any $a_p\in\{a_{i+l},a_{i+l+1},\ldots,a_{i+2l-2}=a_{i-2}\}$. Clearly, $a_p\notin N_D^+(a_i)$. Note that for any choice of $a_p$, the vertex $a_{p+l+1}\in N^+_D(a_i)\cap V(C)$. By Lemma~\ref{cycle_cross_mod}\ref{cross}, we have that $a_p\miss a_{p+l}$. Now applying Corollary~\ref{propcycle}\ref{inoutspecial} to $a_p$, we have that $a_{p+l+1}\in N^-_D(a_p)\cap V(C)$. This gives us that $a_p\in N^{++}_D(a_i)\cap\Gamma(Q)$ for each choice of $a_p\in\{a_{i+l},a_{i+l+1},\ldots,a_{i+2l-2}=a_{i-2}\}$, implying that
	$|N^{++}_D(a_i)\cap\Gamma(Q)|\geq l-1$. Noting that the vertex $a_{i-1}\notin N^{++}_D(a_i)$ (as $a_{i-1}\twoheadrightarrow a_i$), we can now conclude $|N^{++}_D(a_i)\cap\Gamma(Q)|=l-1=|N^+_D(a_i)\cap\Gamma(Q)|$.
\end{proof}
\subsubsection*{Unforced and singly-forced missing edges}
We now label some missing edges of $D$ as unforced and some others as singly-forced.
\begin{definition}\label{singly_forced}
	A missing edge $e = a\miss b$ is said to be \emph{singly-forced} if exactly one of the following conditions hold.
	\vspace{-0.04in}
	\begin{enumerate}
		\itemsep -0.04in
		\renewcommand{\theenumi}{\textup{($\arabic{enumi}$)}}
		\renewcommand{\labelenumi}{\theenumi}
		\item There exists $v\in V(D)$ such that $b\twoheadrightarrow v\rightarrow a$ in $D$.
		\item There exists $u\in V(D)$ such that $a\twoheadrightarrow u\rightarrow b$ in $D$.
	\end{enumerate}
	If $(1)$ holds then we say that $e$ is \emph{forced in the direction $b$ to $a$}, and if $(2)$ holds then we say that $e$ is \emph{forced in the direction $a$ to $b$}. If neither $(1)$ nor $(2)$ hold, then $e$ is \emph{unforced}. Note that it is possible for a missing edge to be forced in both directions.
\end{definition}
\begin{lemma}\label{duallyforced}
	Let $e=a\miss b$. If there exist $u,v\in V(D)$ such that $b\twoheadrightarrow v\rightarrow a$ and $a\twoheadrightarrow u\rightarrow b$, then $(u,v)R(b,a)$. Consequently, if any missing edge is forced in both directions in $D$, then it has an in-neighbor in $\Delta_D$. 
\end{lemma}
\begin{proof}
	Note that $u\neq v$. Now, if $v\rightarrow u$ or $u\rightarrow v$, then $u\rightarrow b\twoheadrightarrow v\rightarrow u$ or $v\rightarrow a\twoheadrightarrow u\rightarrow v$ would form a directed triangle containing a special arc, which is a contradiction. Therefore, $u\miss v$. Then, the fact that $u\rightarrow b\twoheadrightarrow v\rightarrow a\twoheadrightarrow u$ implies that $(u,v)R(b,a)$ and hence $\{u,v\}$ is an in-neighbor of $e$ in $\Delta_D$. \end{proof}

\begin{lemma}\label{strategy}
	Every singly-forced missing edge is the starting vertex of some path in $\mathcal{P}$.
\end{lemma}
\begin{proof}
	Let $a\miss b$ be a singly-forced missing edge. It is enough to prove that $\{a,b\}$ doesn't have any in-neighbor in $\Delta_D$. Assume to the contrary that $\{a,b\}$ has an in-neighbor $\{c,d\}$ in $\Delta_D$. Then by definition of $\Delta_D$ we can assume without loss of generality that $(c,d)R(a,b)$, i.e., there exists a cycle $c\rightarrow a\twoheadrightarrow d\rightarrow b\twoheadrightarrow c$ in $D$. Note that now we have both $b\twoheadrightarrow c \rightarrow a$ and $a\twoheadrightarrow d \rightarrow b$, implying that both conditions (1) and (2) of Definition~\ref{singly_forced} hold. This contradicts  the fact that $a\miss b$ is a singly-forced missing edge.
\end{proof}
\subsubsection*{Completions and special in-neighbors}
A tournament $T$ is said to be a \emph{completion} of $D$ if $V(D)=V(T)$ and $E(D)\subseteq E(T)$. It is easy to see that a completion of $D$ can be obtained by ``orienting'' every missing edge of $D$, i.e., by adding an oriented edge in place of each missing edge of $D$.
When a completion $T$ of $D$ is specified, a missing edge $a\miss b$ of $D$ that has been oriented from $a$ to $b$ in $T$ is denoted by $a\dashrightarrow b$.

\begin{definition}\label{property}
	Given a completion $T$ of $D$ and a vertex $v\in V(T)$, we say that an in-neighbor $b$ of $v$ is a \emph{special in-neighbor} if $b\twoheadrightarrow v$ and $b\in N^{++}_T(v)$. Further, we say that a special in-neighbor $b$ of $v$ is of \emph{Type-I} if there exists $a\in V(T)$ such that $v\rightarrow a\dashrightarrow b\twoheadrightarrow v$. Similarly, we say that a special in-neighbor $b$ of $v$ is of \emph{Type-II} if there exists $a\in V(T)$ such that $v\dashrightarrow a\rightarrow b\twoheadrightarrow v$. Note that any special in-neighbor of $v$ is either Type-I or Type-II or both.
\end{definition}

\begin{lemma}\label{good_vetices}
	Let $T$ be a completion of $D$. For any vertex $v\in V(T)$, every vertex in $N^{++}_T(v)\setminus N^{++}_D(v)$ is a special in-neighbor of $v$.
\end{lemma}
\begin{proof}
	Consider $x\in N^{++}_T(v)\setminus N^{++}_D(v)$. As $x\in N^{++}_T(v)$, $x\in N^-_T(v)$, implying that we have either $x\rightarrow v$ or $x\dashrightarrow v$. Furthermore, there exists $a\in V(T)$ such that $a\in N^+_T(v)\cap N^-_T(x)$. Since $x\notin N^{++}_D(v)$, we know that either $v\dashrightarrow a$ or $a\dashrightarrow x$. As the missing edges of $D$ form a matching, this implies that $x\rightarrow v$. Again using the fact that $x\notin N^{++}_D(v)$, we conclude that $x\twoheadrightarrow v$. This shows that $x$ is a special in-neighbor of $v$.
\end{proof}

\begin{lemma}\label{structure1}
	Let $T$ be a completion of $D$ and $L$ a median order of $T$ such that the feed vertex $d$ of $L$ does not have a special in-neighbor of Type-I. Then $d$ is a vertex with large second neighborhood in $D$.
\end{lemma}
\begin{proof}
	We claim that there exists a completion $T'$ of $D$ such that $L$ is a median order of $T'$ and $d$ has no special in-neighbors in $T'$. If there does not exist a vertex $a\in V(T)$ such that $d\dashrightarrow a$, then clearly $T'=T$ is a completion of $D$ satisfying our requirements. So we shall assume that there exists $a\in V(T)$ with $d\dashrightarrow a$. Now, consider the completion $T'$ of $D$ obtained from $T$ by reorienting the missing edge $d\dashrightarrow a$ as $a\dashrightarrow d$. 
	By Proposition~\ref{reversing}, $L$ is a median order of $T'$ as well. Further, it can be easily seen that $d$ does not have any special in-neighbors of Type-I in $T'$ either. As the only missing edge incident on $d$ is oriented towards $d$ in $T'$, $d$ does not have any special in-neighbors of Type-II in $T'$. This proves our claim.
	
	By Lemma~\ref{good_vetices} applied on $T'$ and $L$, we have $N^{++}_{T'}(d) \subseteq N^{++}_D(d)$. By Theorem~\ref{thmht}, $|N^+_D(d)| = |N^+_{T'}(d)| \leq |N^{++}_{T'}(d)|$ (the first equality is because $a\dashrightarrow d$ in $T'$). Combining this with the previous observation, we have $|N^+_D(d)|\leq |N^{++}_D(d)|$.
\end{proof}
\subsubsection*{Safe completions}

\begin{definition}
	A completion $T$ of $D$ is said to be \emph{safe}
	if it satisfies the following two conditions:
	\vspace{-0.04in}
	\begin{enumerate}
	\itemsep -0.04in
	\renewcommand{\theenumi}{\textup{($\arabic{enumi}$)}}
	\renewcommand{\labelenumi}{\theenumi}
		\item If $a\miss b$ is a singly-forced missing edge that is forced in the direction from $a$ to $b$, then $a\dashrightarrow b$ in $T$, and
		\item if $a\miss b$ is a missing edge such that $\{a,b\}$ does not lie in any cycle in $\mathcal{C}$, $(c,d)R(a,b)$ and $c\dashrightarrow d$ in $T$, then $a\dashrightarrow b$ in $T$.
	\end{enumerate}
\end{definition}

Recall that $(c,d)R(a,b)$ if and only if $(d,c)R(b,a)$. Therefore, if $\{a,b\},\{c,d\}$ are two missing edges that do not lie on any cycle in $\mathcal{C}$ and $(c,d)R(a,b)$, then in any safe completion, $c\dashrightarrow d$ if and only if $a\dashrightarrow b$. The following is an easy consequence of Lemma~\ref{strategy}.


\begin{remark}\label{safeexists}
	Every oriented graph whose missing edges form a matching has a safe completion.
\end{remark}

\begin{lemma}\label{uniqueab}
	Let $T$ be a safe completion of $D$. Let $v\in V(T)$ and $b$ be a Type-I special in-neighbor of $v$. Then there exist $a,u\in V(T)$ such that $v\rightarrow a\dashrightarrow b\twoheadrightarrow v$, $a\twoheadrightarrow u\rightarrow b$ and $u\miss v$. Moreover, $b$ is the only Type-I special in-neighbor of $v$.
\end{lemma}
\begin{proof}
	As $b$ is a Type-I special in-neighbor of $v$, there exists $a\in V(T)$ such that $v\rightarrow a\dashrightarrow b\twoheadrightarrow v$ in $T$. Then by Definition~\ref{singly_forced}, $a\miss b$ is forced in the direction $b$ to $a$. But as we have $a\dashrightarrow b$ in $T$, and every singly-forced missing edge of $D$ was oriented in $T$ in the direction in which it was forced (as $T$ is a safe completion), it must be the case that $a\miss b$ is also forced in the direction $a$ to $b$. That is, there exists $u\in V(T)$ such that $a\twoheadrightarrow u\rightarrow b$ (refer Definition~\ref{singly_forced}). Using Lemma~\ref{duallyforced}, we can now conclude that $(u,v)R(b,a)$, which further implies that $u\miss v$. If there exists a Type-I special in-neighbor $b'$ of $d$ such that $b'\neq b$, then the same arguments can be used to infer that there exist $a',u'\in V(T)$ such that $(u',v)R(b',a')$ (which means that $u'\miss v$). Since $v$ has at most one non-neighbor, we have that $u'=u$, which gives $(u,v)R(b',a')$. As it can be easily seen that $\{a',b'\}\neq\{a,b\}$, the missing edge $\{u,v\}$ has more than one out-neighbor in $\Delta_D$, which is a contradiction to Lemma~\ref{indegree}. Hence $b$ is the only Type-I special in-neighbor of $v$.
\end{proof}
\begin{lemma}\label{outmissing}
	Let $T$ be a safe completion of $D$ and let $L$ be a median order of $T$ with feed vertex $d$. If $d$ has a Type-I special in-neighbor $b$ and there exists $w\in V(T)$ such that $d\dashrightarrow w$, then:
	\begin{myenumerate}
		\item\label{claim} $N^{++}_T(d)\setminus\{b\}\subseteq N^{++}_D(d)$, and
		\item\label{disgood} $d$ is a vertex with large second neighborhood in $D$.
	\end{myenumerate}
\end{lemma}
\begin{proof}
	By Lemma~\ref{uniqueab}, there exist $a,u\in V(T)$ such that $d\rightarrow a\dashrightarrow b\twoheadrightarrow d$, $a\twoheadrightarrow u\rightarrow b$ and $u\miss d$. As the only non-neighbor of $d$ is $w$, we have $u=w$.
	
	\ref{claim} Consider a vertex $x\in N^{++}_T(d)\setminus\{b\}$. Suppose for the sake of contradiction that $x\notin N^{++}_D(d)$. Then by Lemma~\ref{good_vetices}, we know that $x$ is a special in-neighbor of $d$. Since $x\neq b$, we know by Lemma~\ref{uniqueab} that $x$ cannot be a Type-I special in-neighbor of $d$. Therefore, $x$ is a Type-II special in-neighbor of $d$, i.e., $d\dashrightarrow w\rightarrow x\twoheadrightarrow d$ (as $w$ is the only non-neighbor of $d$). It is easily verified that $a\neq x$. Further, $\{a,x\}$ cannot be a missing edge since $a\miss b$ and $x\neq b$. If $x\rightarrow a$ or $a\rightarrow x$, then either $a\twoheadrightarrow u=w\rightarrow x\rightarrow a$ or $d\rightarrow a\rightarrow x\twoheadrightarrow d$ would be a directed triangle containing a special arc, which is a contradiction. This proves~\ref{claim}.
	
	\ref{disgood} We have $|N_D^+(d)|= |N_T^+(d)| - 1\leq |N^{++}_T(d)|-1 = |N^{++}_T(d)\setminus \{b\}|\leq |N^{++}_D(d)|$ (the first equality is because $d\dashrightarrow w$, the second inequality by Theorem~\ref{thmht}, the third equality is because $b\in N^{++}_T(d)$, and the fourth inequality by~\ref{claim}).
\end{proof}

Consider a module $I$ in $D$ such that $|I|\geq 2$ and a vertex $v\in I$.
Clearly, any non-neighbor of $v$ outside $I$ has to be a non-neighbor of every vertex in $I$. As $|I|\geq 2$ and the missing edges of $D$ form a matching, this can only mean that $v$ has no non-neighbors outside $I$. We thus have the following observation.

\begin{remark}\label{missinginmod}
If $I$ is a module in $D$ such that $|I|\geq 2$ and $v\in I$, then $v$ has no non-neighbors outside $I$.
\end{remark}

\begin{lemma}\label{feedcycle}
Let $T$ be a safe completion of $D$ and let $I$ be a module in $D$ with $|I|\geq 2$. Then for any $v\in I$, $N^{++}_T(v)\setminus I \subseteq N^{++}_D(v)\setminus I$.
\end{lemma}
\begin{proof}
First, suppose that there exists a Type-I special in-neighbor $x$ of $v$ outside $I$. By Lemma~\ref{uniqueab}, there exists a vertex $u\in V(T)$ such that $u\miss v$ and $u\rightarrow x$. By Remark~\ref{missinginmod}, $u\in I$. Now we have $x\rightarrow v$ and $u\rightarrow x$, which contradicts the fact that $u$ and $v$ belong to the module $I$ in $D$ and $x$ is outside that module. Therefore, $v$ has no Type-I special in-neighbors outside $I$. 
Next, suppose that there exists a Type-II special in-neighbor $x$ of $v$ outside $I$. Then, there exists a vertex $y$ such that $v\dashrightarrow y\rightarrow x\twoheadrightarrow v$. By Remark~\ref{missinginmod}, we know that $y\in I$. Then we have $x\rightarrow v$ and $y\rightarrow x$, which contradicts the fact that $v$ and $y$ belong to the module $I$ in $D$ (recall that $x$ is outside $I$). Therefore, we can conclude that $v$ has no special in-neighbors outside $I$. This implies, by Lemma~\ref{good_vetices}, that $N^{++}_T(v)\setminus I\subseteq N^{++}_D(v)\setminus I$.
\end{proof}

\begin{corollary}\label{mainresult}
Let $T$ be a safe completion of $D$ and let $d$ be the feed vertex of some median order of $T$. Let $I$ be a module in $D$ containing $d$ where $|I|\geq 2$. Then for any $v\in I$, $|N^+_D(v)\setminus I|\leq |N^{++}_D(v)\setminus I|$.
\end{corollary}
\begin{proof}
It is easy to see that as the missing edges of $D$ form a matching, every module in $D$ is also a module in $T$. Therefore $I$ is a module in $T$ containing $d$. Then we have from Proposition~\ref{modht} and Lemma~\ref{feedcycle} that $|N^+_D(v)\setminus I|\leq |N^+_T(v)\setminus I|\leq |N^{++}_T(v)\setminus I|\leq |N^{++}_D(v)\setminus I|$.
\end{proof}
\subsubsection*{Prime vertices}
We define $$I(u)=\left\{\begin{array}{ll}\Gamma(Q)&\mbox{if }\exists Q\in\mathcal{C}\mbox{ such that }u\in\Gamma(Q)\mbox{,}\\\{u\}&\mbox{otherwise.}\end{array}\right.$$ Note that as any vertex $u$ can be a part of at most one missing edge, there can be at most one cycle $Q\in\mathcal{C}$ such that $u\in\Gamma(Q)$, and therefore $I(u)$ is well defined. We define a vertex $u$ in $D$ to be \emph{prime}, if $I(u)=\{u\}$; in other words, a vertex $u$ is said to be prime if $u\notin\Gamma(Q)$ for any $Q\in\mathcal{C}$.

Note that if $u$ is prime, we have $I(u)=\{u\}$ and therefore, $|N^+_D(u)\cap I(u)|=|N^{++}_D(u)\cap I(u)|=0$. On the other hand, if $u\in\Gamma(Q)$ for some $Q\in\mathcal{C}$, then $I(u)=\Gamma(Q)$, and by Lemma~\ref{snc_cycle}, we get that $|N^+_D(u)\cap I(u)|=|N^{++}_D(u)\cap I(u)|$. We thus have the following.

\begin{remark}\label{sncmodule}
	For any vertex $u\in V(D)$, $|N^+_D(u)\cap I(u)|=|N^{++}_D(u)\cap I(u)|$.
\end{remark}
The following result is implicit in the work of Ghazal~\cite{ghazal2015remark}.

\begin{theorem}\label{mainthm}
Let $d$ be the feed vertex of some median order of a safe completion $T$ of $D$. Then every vertex in $I(d)$ has a large second neighborhood in $D$.
\end{theorem}
\begin{proof}
Suppose that $d$ is prime. Then, $I(d)=\{d\}$. If $d$ has no special in-neighbor of Type-I in $T$, then we are done by Lemma~\ref{structure1}. So let us assume that $d$ has a special in-neighbor $b$ of Type-I in $T$. Then by Lemma~\ref{uniqueab}, there exist $a,u\in V(T)$ such that $d\rightarrow a\dashrightarrow b\twoheadrightarrow d$, $a\twoheadrightarrow u\rightarrow b$, where $u\miss d$. This means that $(u,d)R(b,a)$. If $u\dashrightarrow d$, then since $T$ is a safe completion of $D$, the fact that $a\dashrightarrow b$ implies that $\{u,d\}$ and $\{b,a\}$ lie in some cycle in $\mathcal{C}$, contradicting the assumption that $d$ is prime. Therefore, we have $d\dashrightarrow u$. Then we are done by Lemma~\ref{outmissing}\ref{disgood}.

Next, consider the case when $d$ is not prime, i.e. $d\in\Gamma(Q)$ for some $Q\in\mathcal{C}$. Note that we then have $I(d)=\Gamma(Q)$ and therefore, $|I(d)|\geq 2$.	Consider any vertex $v\in I(d)$. As $I(d)=\Gamma(Q)$ is a module (Corollary~\ref{gamma}\ref{gammamodule}), we have by Corollary~\ref{mainresult} that $|N^+_D(v)\setminus I(d)|\leq |N^{++}_D(v)\setminus I(d)|$. By Remark~\ref{sncmodule}, $|N^+_D(v)\cap I(d)|= |N^{++}_D(v)\cap I(d)|$. We now have $|N^+_D(v)| = |N^+_D(v)\setminus I(d)|+ |N^+_D(v)\cap I(d)| \leq |N^{++}_{D}(v)\setminus I(d)|+|N^{++}_D(v)\cap I(d)|=|N^{++}_D(v)|$. Hence the theorem.
\end{proof}


\subsubsection*{Reverse special arcs}
We now state a property of special arcs that are ``reverse arcs'' in a median order, which will be useful for deriving the results in the next section.
\begin{definition}
	Given a median order $L=(x_1,x_2,\ldots,x_n)$ of any completion $T$ of $D$, a special arc $x_j\twoheadrightarrow x_i$ is said to be a \emph{reverse special arc} in $(T,L)$ if $i<j$.
\end{definition}
\begin{lemma}\label{rev_specl_arc}
	Let $L=(x_1,x_2,\ldots,x_n)$ be a median order of a completion $T$ of $D$ and $x_j\twoheadrightarrow x_i$ be a reverse special arc in $(T,L)$. Then at least one of the following conditions hold:
	\begin{myenumerate}
		\item\label{mis_out} There exists $x_k$ such that $x_i\dashrightarrow x_k\rightarrow x_j$, where $i<k<j$, or
		\item\label{mis_in} There exists $x_l$ such that $x_i\rightarrow x_l\dashrightarrow x_j$, where $i<l<j$.
	\end{myenumerate}
	Moreover, if exactly one of the above conditions holds,  then $L'= (x_1,x_2,\ldots, x_{i-1},x_{i+1},\ldots,x_j,x_i,x_{j+1},\ldots,$ $ x_n)$ is also a median order of $T$.
\end{lemma}
\begin{proof}
	For the purposes of this proof, for $u\in\{x_i,x_{i+1},\ldots,x_j\}$, we shall abbreviate $N^+_T(u)\cap\{x_i,x_{i+1},\ldots,x_j\}$ and $N^-_T(u)\cap\{x_i,x_{i+1},\ldots,x_j\}$ to just $N^+_{i,j}(u)$ and $N^-_{i,j}(u)$ respectively.
	By Lemma~\ref{feedback}, we have
	\begin{equation}\label{local_x_i}
		\left|N^+_{i,j}(x_i)\right| \geq \frac{j-i}{2}\quad \text{and}\quad \left |N^-_{i,j}(x_j)\right| \geq \frac{j-i}{2}
	\end{equation}
	Alternatively,
	\begin{equation}\label{local_x_j}
		\left|N^-_{i,j}(x_i)\right| \leq \frac{j-i}{2}\quad \text{and}\quad \left|N^+_{i,j}(x_j)\right| \leq \frac{j-i}{2}
	\end{equation}
	
	We shall first make an observation about any vertex $x_p\in N^+_{i,j}(x_i)\setminus N^+_{i,j}(x_j)$. Clearly, $x_p\in N^+_{i,j}(x_i)\cap N^-_{i,j}(x_j)$ (recall that $x_j\twoheadrightarrow x_i$). Note that either $x_i\dashrightarrow x_p$ or $x_p\dashrightarrow x_j$, as otherwise $x_i\rightarrow x_p\rightarrow x_j\twoheadrightarrow x_i$ would form a directed triangle in $D$ containing a special arc, which is a contradiction. Since the missing edges of $D$ form a matching, this implies that either $x_i\dashrightarrow x_p\rightarrow x_j$ or $x_i\rightarrow x_p\dashrightarrow x_j$. 
	
	Suppose that neither of the conditions in the lemma hold. Then from the above observation, it is clear that $N^+_{i,j}(x_i)\subseteq N^+_{i,j}(x_j)$. Note that $x_i\notin N^+_{i,j}(x_i)$ but $x_i\in N^+_{i,j}(x_j)$. Therefore we have, $\left|N^+_{i,j}(x_i)\right|< \left|N^+_{i,j}(x_j)\right|\leq \frac{j-i}{2}$ (by \eqref{local_x_j}), which contradicts \eqref{local_x_i}. Therefore at least one of the conditions \ref{mis_out} or \ref{mis_in} should hold.
	
	Now suppose that exactly one of the conditions~\ref{mis_out} or~\ref{mis_in} holds.
	Note first that from the previous observation and the fact that the missing edges of $D$ form a matching, it follows that if there exist two distinct vertices $x_p,x_q$ in $N^+_{i,j}(x_i)\setminus N^+_{i,j}(x_j)$, then $x_i\dashrightarrow x_p\rightarrow x_j$ and $x_i\rightarrow x_q\dashrightarrow x_j$, implying that both conditions hold. Therefore, there is exactly one vertex in $N^+_{i,j}(x_i)\setminus N^+_{i,j}(x_j)$, i.e., $|N^+_{i,j}(x_i)\setminus N^+_{i,j}(x_j)|=1$. Since $x_i\in N^+_{i,j}(x_j)\setminus N^+_{i,j}(x_i)$, we have that $|N^+_{i,j}(x_i)\setminus (N^+_{i,j}(x_j)\setminus\{x_i\})|=1$. This means that $|N^+_{i,j}(x_i)|-(|N^+_{i,j}(x_j)|-1)\leq 1$, implying that $|N^+_{i,j}(x_i)|\leq |N^+_{i,j}(x_j)|$.
	Hence, $\frac{j-i}{2} \leq \left|N^+_{i,j}(x_i)\right|\leq \left|N^+_{i,j}(x_j)\right|\leq \frac{j-i}{2}$ (from \eqref{local_x_i} and \eqref{local_x_j}). Therefore, we have $\left|N^+_{i,j}(x_i)\right|= \frac{j-i}{2}= \left|N^-_{i,j}(x_i)\right|$ (which means that $j-i$ is even). Then by Lemma~\ref{movexi}\ref{xi}, $L'= (x_1,x_2,\ldots, x_{i-1},x_{i+1},\ldots,x_j,x_i,x_{j+1},\ldots, x_n)$ is also a median order of $T$.
\end{proof}

We now prove another lemma that will be needed later.

\begin{lemma}\label{noin-neighbor}
Let $T$ be a safe completion of $D$ and let $L$ be a median order of $T$ whose feed vertex $d$ is prime. Suppose that there exists $d'\in V(T)$ such that $d\dashrightarrow d'$ in $T$. Then the missing edge $\{d,d'\}$ does not have an in-neighbor in $\Delta_D$.
\end{lemma}
\begin{proof}
Suppose not. Let $\{a,a'\}$ be an in-neighbor of $\{d,d'\}$ in $\Delta_D$. Then, without loss of generality, by the definition of $\Delta_D$, we can assume that $(a,a')R(d,d')$, and therefore there exists the four cycle $a\rightarrow d\twoheadrightarrow a'\rightarrow d'\twoheadrightarrow a$.  As $d$ is prime, $d$ does not belong to $\Gamma(Q)$ for any $Q\in\mathcal{C}$. This means that $\{d,d'\}$ does not lie in any cycle in $\mathcal{C}$. Then as $T$ is a safe completion, $d\dashrightarrow d'$ and $(a,a')R(d,d')$, we have that $a\dashrightarrow a'$ in $T$. As $d$ is the feed vertex of $L$, $d\twoheadrightarrow a'$ is a reverse special arc in $(T,L)$. Note that the only missing edge $d\miss d'$ incident on $d$ is oriented as $d\dashrightarrow d'$, and the only missing edge $a\miss a'$ incident on $a'$ is oriented as $a\dashrightarrow a'$. This implies that neither of the conditions \ref{mis_out} or \ref{mis_in} of Lemma~\ref{rev_specl_arc} hold, which is a contradiction.
\end{proof}


\subsection{Tournaments missing a matching and a star}\label{sec:nearmatching}

In this section, we shall show that if the missing edges of an oriented graph can be partitioned into a matching and a star, then it contains a vertex with a large second neighborhood. As noted in the beginning, any sink in an oriented graph is a vertex with a large second neighborhood. Therefore, we only need to show the result for graphs that contain no sink. In fact, we show the following stronger result.

\begin{theorem}\label{finalthm}
Let $H$ be an oriented graph that does not contain a sink and $z\in V(H)$ such that $D=H-\{z\}$ is a tournament missing a matching. Then there exists a vertex in $V(D)$ that has a large second neighborhood in both $D$ and $H$.
\end{theorem}

When $H$ is a tournament missing a matching and a star, and $H$ does not contain a sink, we can apply the above theorem taking $z$ to be the center of the star, to obtain the result that there is a vertex other than $z$ having a large second neighborhood in $H$.

For the remainder of this section, we assume that $H$ is an oriented graph without a sink containing a vertex $z\in V(H)$ such that $D=H-\{z\}$ is a tournament missing a matching.

\begin{lemma}\label{twovertices}
Let $d\in V(D)$ be a vertex that has a large second neighborhood in $D$. If $d$ does not have a large second neighborhood in $H$, then:
\begin{myenumerate}
\item\label{xnoutneigh} $z\in N^+_H(d)$, and
\item\label{xjbadvertex} $N^+_H(z)\subseteq N^+_D(d)\cup N^{++}_D(d)$.
\end{myenumerate}
\end{lemma}
\begin{proof}
	Since $d$ has a large second neighborhood in $D$, we have $|N^+_D(d)|\leq |N^{++}_D(d)|$.\smallskip

	\ref{xnoutneigh} If $z\notin N^+_H(d)$, then since $N^{++}_D(d)\subseteq N^{++}_H(d)$, we have $|N^+_H(d)| = |N^+_D(d)|\leq |N^{++}_D(d)|\leq |N^{++}_D(d)|$. This contradicts the assumption that $d$ does not have a large second neighborhood in $H$.\smallskip
	
	\ref{xjbadvertex} Suppose for the sake of contradiction that there exists $u\in N^+_H(z)\setminus (N^+_D(d)\cup N^{++}_D(d))$. From~\ref{xnoutneigh}, $z\in N^+_H(d)$. As $z\in N^+_H(d)\cap N^-_H(u)$ and $u\notin N^+_D(d)$, we get $u\in N^{++}_H(d)$. 
	Combining all these together we get,
	\begin{eqnarray*}
		|N^+_H(d)| &=& |N^+_D(d)|+1\hspace{.25in}\mbox{ (as $z\in N^+_H(d)$)}\\
		&\leq& |N^{++}_D(d)|+1\\
		&=& |N^{++}_D(d)\cup\{u\}|\hspace{.25in}\mbox{ (as $u\notin N^{++}_D(d)$)}\\
		&\leq& |N^{++}_H(d)|\hspace{.25in}\mbox{ (as $N^{++}_D(d)\subseteq N^{++}_H(d)$ and $u\in N^{++}_H(d)$)}
	\end{eqnarray*} 
	and therefore $d$ has a large second neighborhood in $H$, which is a contradiction.
\end{proof}

Define $\mathcal{I}(D)=\{I(u)\colon u\in V(D)\}=\{\Gamma(Q)\colon Q\in\mathcal{C}\}\cup\{\{u\}\colon u$ is prime$\}$. By Corollary~\ref{gamma}\ref{gammamodule}, $\mathcal{I}(D)$ is a partition of $V(D)$ into modules of $D$. It is easy to verify that since the missing edges of $D$ form a matching, every module in $D$ is also a module in any completion $T$ of $D$. This implies that $\mathcal{I}(D)$ is a partition of $V(T)$ into modules of $T$ as well. Therefore, by Lemma~\ref{module}, there exists a good median order of $T$ with respect to $\mathcal{I}(D)$.
In fact, Lemma~\ref{module} gives the following stronger observation.
It is easy to see that as the missing edges of $D$ form a matching, every module in $D$ is also a module in $T$. Since by Corollary~\ref{gamma}\ref{gammamodule}, for each $Q\in\mathcal{C}$, $\Gamma(Q)$ is a module in $D$, we have the following observation.
\begin{remark}\label{gammaT}
For each $Q\in\mathcal{C}$, $\Gamma(Q)$ is a module in $T$.
\end{remark}

\begin{remark}\label{goodorder}
	If $L$ is a median order of any completion $T$ of $D$, then $T$ has a good median order with respect to $\mathcal{I}(D)$ with the same feed vertex as $L$.
\end{remark}

\begin{lemma}\label{periodic}
	Let $T$ be a safe completion of $D$.
	If there exists a good median order $L$ of $T$ with respect to $\mathcal{I}(D)$ which is periodic, then there exists $x\in V(D)$ such that $x$ has a large second neighborhood in both $D$ and $H$.
\end{lemma}
\begin{proof}
    For the purposes of this proof, given an ordering of vertices $\hat{L}=(x_1,x_2,\ldots,x_n)$ and a vertex $v\in\{x_1,x_2,\ldots,x_n\}$, we define the ``index of $v$ in $\hat{L}$'' to be the integer $i$ such that $x_i=v$.

	Let us denote the feed vertex of $Sed^i_{\mathcal{I}(D)}(L)$ by $d_i$. 
	By Theorem~\ref{sedimentation}, we know that for any integer $i\geq 0$, $Sed^i_{\mathcal{I}(D)}(L)$ is a good median order of $T$ with respect to $\mathcal{I}(D)$.
	 
	As $H$ does not have any sink, there exists $v\in V(D)$ such that $v\in N^+_H(z)$. Suppose that there exists an integer $i\geq 0$ such that $v\in I(d_i)$. By Theorem~\ref{mainthm}, $v$ has a large second neighborhood in $D$. Further, as $z\notin N^+_H(v)$, by Lemma~\ref{twovertices}\ref{xnoutneigh}, $v$ has a large second neighborhood in $H$ too, and we are done. Thus we can assume that $v\notin I(d_i)$ for any integer $i\geq 0$. Since this implies that $v\neq d_i$ for any positive integer $i\geq 0$, there exists an integer $q\geq 0$ such that the index of $v$ in $Sed^{q+1}_{\mathcal{I}(D)}(L)$ is less than or equal to its index in $Sed^q_{\mathcal{I}(D)}(L)$ (recall that $L$ is periodic). Since $v\notin I(d_q)$, this means that $v\in N^-_T(d_q)\setminus N^{++}_T(d_q)$, which implies that $v\notin N^+_D(d_q)\cup N^{++}_D(d_q)$. Notice that by Theorem~\ref{mainthm}, $d_q$ has a large second neighborhood in $D$. Since $v\in N^+_H(z)\setminus (N^+_D(d_q)\cup N^{++}_D(d_q))$, we have by Lemma~\ref{twovertices}\ref{xjbadvertex} that $d_q$ has a large second neighborhood in $H$ too. Thus the vertex $d_q$ has a large second neighborhood in both $D$ and $H$.
\end{proof}

By the above lemma, if any safe completion of $D$ has a good median order with respect to $\mathcal{I}(D)$ that is periodic, then we are done. In order to complete the proof of Theorem~\ref{finalthm}, we shall show that if a special kind of safe completion $T$ of $D$ has some median order that is stable, then again there will exist a vertex that has a large second neighborhood in both $D$ and $H$. The remainder of this section is devoted proving this fact.
By the above lemma, henceforth we can focus our attention on the case when for any safe completion $T$ of $D$, every good median order of $T$ with respect to $\mathcal{I}(D)$ is stable. That is, if $L$ is a good median order of a safe completion $T$ of $D$ with respect to $\mathcal{I}(D)$, there exists an integer $q\geq 0$ such that the feed vertex $d$ of the median order $Sed^q_{\mathcal{I}(D)}(L)$ satisfies $|N^{++}_T(d)\setminus I(d)|> |N^+_T(d)\setminus I(d)|$. Therefore, to complete the proof of Theorem~\ref{finalthm}, we only need to show that if there exists some safe completion $T$ of $D$ having a median order with feed vertex $d$ such that $|N^{++}_T(d)\setminus I(d)|> |N^+_T(d)\setminus I(d)|$, then $d$ has a large second neighborhood in $H$ (by Theorem~\ref{mainthm}, it anyway has a large second neighborhood in $D$). If $D$ has no missing edges, then since $T=D$ and $I(d)=\{d\}$, it is straightforward to see that $d$ has a large second neighborhood in $H$. But to prove that $d$ has a large second neighborhood in $H$ even if $D$ contains some missing edges requires some more work. The remainder of the section is devoted to proving this fact, which we state as Lemma~\ref{stable2}.
\medskip

\begin{lemma}\label{stable1}
Let $T$ be a safe completion of $D$ and let $L$ be a median order of $T$ having feed vertex $d$ such that $|N^{++}_T(d)\setminus I(d)|> |N^+_T(d)\setminus I(d)|$. If either $d$ has no special in-neighbors or $d$ has a special in-neighbor of Type-I, then $d$ has a large second neighborhood in $H$.
\end{lemma}
\begin{proof}
By Theorem~\ref{mainthm}, $d$ has a large second neighborhood in $D$.
If $z\notin N^+_H(d)$, then we are done by Lemma~\ref{twovertices}\ref{xnoutneigh}. So we can assume that $z\in N^+_H(d)$.
	
Suppose that $d$ has no special in-neighbors. Then, by Lemma~\ref{good_vetices}, we have $N^{++}_T(d)\subseteq N^{++}_D(d)$. Consequently, $N^{++}_T(d)\setminus I(d)\subseteq N^{++}_D(d)\setminus I(d)$.  
	
	Now suppose that $d$ has a special in-neighbor of Type-I and $d$ is not prime. Then there exists $Q\in \mathcal{C}$ such that $I(d)=\Gamma(Q)$. By Lemma~\ref{feedcycle}, we have $N^{++}_T(d)\setminus I(d)\subseteq N^{++}_D(d)\setminus I(d)$.
	
	Therefore, if $d$ has no special in-neighbors or if $d$ has a special in-neighbor of Type-I but $d$ is not prime, we have $N^{++}_T(d)\setminus I(d)\subseteq N^{++}_D(d)\setminus I(d)$. In that case, we get,
	\begin{eqnarray*}
		|N^+_H(d)| &=& |N^+_D(d)|+1\quad (\text{as }z\in N^+_H(d))\\
		&=& |N^+_D(d)\setminus I(d)|+ |N^+_D(d)\cap I(d)|+1 \\
		&\leq& |N^+_T(d)\setminus I(d)|+ |N^{++}_D(d)\cap I(d)|+1\quad (\text{since }N^+_D(d)\subseteq N^+_T(d)\text{ and by Remark~\ref{sncmodule}})\\
		&\leq& |N^{++}_T(d)\setminus I(d)| + |N^{++}_D(d)\cap I(d)|\quad (\text{as }|N^{++}_T(d)\setminus I(d)|> |N^+_T(d)\setminus I(d)|)\\
		&\leq& |N^{++}_D(d)\setminus I(d)|+ |N^{++}_D(d)\cap I(d)|\quad (\text{as }N^{++}_T(d)\setminus I(d)\subseteq N^{++}_D(d)\setminus I(d))\\
		&=& |N^{++}_D(d)|\\
		&\leq& |N^{++}_H(d)|
	\end{eqnarray*}
	and hence $d$ has a large second neighborhood in $H$. Now, to prove the lemma, it only remains to consider the case when $d$ has a special in-neighbor $b$ of Type-I and $d$ is prime. Then by Lemma~\ref{uniqueab}, there exist $a,u\in V(T)$ such that $d\rightarrow a\dashrightarrow b\twoheadrightarrow d$ and $a\twoheadrightarrow u\rightarrow b$, where $u\miss d$. This means that $(u,d)R(b,a)$. If $u\dashrightarrow d$, then since $T$ is a safe completion of $D$, the fact that $a\dashrightarrow b$ implies that $\{u,d\}$ and $\{b,a\}$ lie in some cycle $Q$ in $\mathcal{C}$. But then $d\in\Gamma(Q)$, which contradicts the fact that $d$ is prime. Therefore, we have $d\dashrightarrow u$. Then by Lemma~\ref{outmissing}\ref{claim}, we have $N^{++}_T(d)\setminus\{b\} \subseteq N^{++}_D(d)$. Combining all these together, we have 
	\begin{eqnarray*}
		|N^+_H(d)| &=& |N^+_D(d)|+1\quad (\text{as }z\in N^+_H(d))\\
		&=& |N^+_T(d)| -1 +1 \quad(\text{as } d\dashrightarrow u \text{ in }T)\\
		&\leq& |N^{++}_T(d) \setminus\{b\}| \quad (\text{as } I(d)= \{d\} \text{, we have }|N^{++}_T(d)|> |N^+_T(d)|)\\
		&\leq& |N^{++}_D(d)|\quad (\text{as }N^{++}_T(d)\setminus\{b\} \subseteq N^{++}_D(d))\\
		&\leq& |N^{++}_H(d)|
	\end{eqnarray*}
	Hence the lemma.
	\medskip
\end{proof}

\subsubsection*{The relation $F$}
Define a relation $F$ on $V(D)$ as follows. For $x,y\in V(D)$ such that $x$ is prime, we say that $xFy$ if and only if $x\twoheadrightarrow y$ and there exists $x'\in V(D)$ such that $y\rightarrow x'$ and $x \miss x'$ is a singly-forced missing edge in $D$. Note that if $xFy$, then the missing edge $x \miss x'$ is forced in the direction $x$ to $x'$, and the condition that $x\miss x'$ is singly-forced ensures that it is not forced in the direction $x'$ to $x$. 

\begin{lemma}\label{relation_cycle}
	The relation $F$ is not cyclic, i.e. there does not exist vertices $x_1,x_2,\ldots, x_k\in V(D)$ such
	that $x_1 F x_2 F\cdots$ $Fx_k F x_1$.
\end{lemma}
\begin{proof}
Suppose that there exist vertices $x_1,x_2,\ldots, x_k\in V(D)$ such
that $x_1 F x_2 F\cdots$ $Fx_k F x_1$. Then by the definition of $F$, there is a special cycle $C=x_1\twoheadrightarrow x_2\twoheadrightarrow x_3\twoheadrightarrow\cdots\twoheadrightarrow x_k\twoheadrightarrow x_1$. By Corollary~\ref{propcycle}\ref{missonespecial}, we know that there exists $i\in\{3,4,\ldots,k-1\}$ such that $x_1\miss x_i$. Then by the definition of $F$, the fact that $x_1Fx_2$ implies that the missing edge $x_1\miss x_i$ is forced in the direction $x_1$ to $x_i$ and not forced in the direction $x_i$ to $x_1$. But since the only non-neighbor of $x_i$ is $x_1$, the fact that $x_iFx_{i+1}$ similarly implies that the missing edge $x_1\miss x_i$ is forced in the direction $x_i$ to $x_1$, which is a contradiction.
\end{proof}

\begin{lemma}\label{relation_miss}
Let $T$ be a safe completion of $D$ and let $x$ be the feed vertex of a median order $L$ of $T$. Suppose that $x$ is prime and there exists $y\in V(D)$ such that $xFy$. Then:
	\begin{myenumerate}
	\item\label{y'exists} there exists $y'\in V(D)$ such that $y\dashrightarrow y'\rightarrow x$ in $T$,
	\item\label{newmo} there exists a median order of $T$ having feed vertex $y$, and
	\item\label{yprime} $y$ is prime.
	\end{myenumerate} 
\end{lemma}
\begin{proof}
	By the definition of $xFy$, we have that $x\twoheadrightarrow y$ and that there exists $x'\in V(D)$ such that $y\rightarrow x'$, and $x\miss x'$ is a singly-forced missing edge that is forced in the direction from $x$ to $x'$. As $T$ is a safe completion, we then have $x\dashrightarrow x'$. Since $x$ is the feed vertex of $L$, $x\twoheadrightarrow y$ is a reverse special arc in $(T,L)$, where the only missing edge $x\miss x'$ that is incident on $x$ is oriented as $x\dashrightarrow x'$. This implies that condition~\ref{mis_in} of Lemma~\ref{rev_specl_arc} does not hold. Therefore, condition~\ref{mis_out} of the lemma must be true, i.e. there should exist $y'\in V(T)$ such that $y\dashrightarrow y'\rightarrow x$ in $T$. This proves~\ref{y'exists}. Now as exactly one of the conditions of the lemma is satisfied, if $L=(x_1,x_2,\ldots,x_n=x)$ and $y=x_i$, for some $i\in\{1,2,\ldots,n-1\}$, we have by the lemma that $L'= (x_1,x_2,\ldots, x_{i-1},x_{i+1},\ldots, x_n,x_i=y)$ is also a median order of $T$. Thus we have~\ref{newmo}.
	We shall now prove~\ref{yprime}. If $y$ is not prime, then we have that, $y\in \Gamma(Q)$ for some $Q\in \mathcal{C}$. Therefore by Corollary~\ref{gamma}\ref{missingingamma}, $y'\in \Gamma(Q)$. But we have a vertex $x\in V(D)\setminus \Gamma(Q)$ (as $x$ is prime) such that $x\rightarrow y$ and $y'\rightarrow x$. As $y,y'\in\Gamma(Q)$, this contradicts the fact that $\Gamma(Q)$ forms a module in $D$ (by Corollary~\ref{gamma}\ref{gammamodule}). Therefore we can conclude that $y$ is prime.
\end{proof}

\subsubsection*{Safe completions having maximum value}
We say that the \emph{value} of a tournament is the number of forward arcs in any median order of it. A \emph{maximum value safe completion} of $D$ is a safe completion $T$ of $D$ having the largest value among the safe completions of $D$. In other words, $T$ is a safe completion of $D$ with smallest \emph{feedback arc set} (a set of arcs whose removal makes $T$ acyclic). By Remark~\ref{safeexists}, we know that there always exists a maximum value safe completion of $D$.

\begin{lemma}\label{unforced}
Let $T$ be a maximum value safe completion of $D$. If $L=(x_1,x_2,\ldots,x_n)$ is a median order of $T$, then there cannot exist a missing edge $x_j\dashrightarrow x_i$, where $i<j$, such that $\{x_i,x_j\}$ is an isolated vertex of $\Delta_D$ and $x_i\miss x_j$ is unforced.
\end{lemma}
\begin{proof}
	Let $T'$ be the tournament obtained from $T$ by reversing the arc $x_j\dashrightarrow x_i$. By Lemma~\ref{reversing}, $L$ is a median order of $T'$ as well. Also, as $x_i\miss x_j$ is unforced and $\{x_i,x_j\}$ is an isolated vertex in $\Delta_D$, the tournament $T'$ is also a safe completion of $D$. This contradicts the fact that $T$ is a maximum value safe completion as $T'$ has higher value than $T$.
\end{proof}

\begin{lemma} \label{dingamma}
Let $T$ be a maximum value safe completion of $D$. Let $d$ be the feed vertex of a median order $L$ of $T$ and $d'\in V(T)$ be such that $d\dashrightarrow d'$ in $T$. If $d$ has no Type-I special in-neighbor then $d$ is not prime.
\end{lemma} 
\begin{proof}
Suppose for the sake of contradiction that $d$ has no Type-I special in-neighbor and $d$ is prime. Let $(y_1,y_2,\ldots,y_k)$ be a maximum length sequence of vertices in $D$ such that $y_1=d$ and $y_1Fy_2F\cdots Fy_k$ (note that we allow $k=1$, and therefore such a sequence always exists).
Since $F$ is acyclic as shown in Lemma~\ref{relation_cycle}, each vertex in $(y_1,y_2,\ldots,y_k)$ is distinct. By Lemma~\ref{relation_miss}\ref{yprime} and~\ref{newmo}, we know that if for some $i\in\{1,2,\ldots,k-1\}$, $y_i$ is prime and is the feed vertex of a median order $L_i$ of $T$, then $y_{i+1}$ is prime and is the feed vertex of a median order $L_{i+1}$ of $T$. Thus, since $y_1=d$ is prime and is the feed vertex of a median order $L_1=L$ of $T$, we have by induction on $i$ that for each $i\in\{1,2,\ldots,k\}$, $y_i$ is prime and is the feed vertex of a median order $L_i$ of $T$. We first show that $k\geq 2$.

By Lemma~\ref{noin-neighbor}, the missing edge $\{d,d'\}$ does not have an in-neighbor in $\Delta_D$. Now, suppose that $\{d,d'\}$ has an out-neighbor, say $\{a,a'\}$ in $\Delta_D$. Without loss of generality, by the definition of $\Delta_D$, we can assume that $(d,d')R(a,a')$. That is, there exists the four cycle $d\rightarrow a\twoheadrightarrow d'\rightarrow a'\twoheadrightarrow d$. As $d$ is prime, the missing edge $\{d,d'\}$ does not lie on any cycle in $\mathcal{C}$. Since $T$ is a safe completion, $d\dashrightarrow d'$ and $(d,d')R(a,a')$, we have $a\dashrightarrow a'$ in $T$. Then, $d\rightarrow a\dashrightarrow a'\twoheadrightarrow d$, implying that $a'$ is a Type-I special in-neighbor of $d$, which is a contradiction. Therefore, we can conclude that the missing edge $\{d,d'\}$ is an isolated vertex in $\Delta_D$. Then, since $d\dashrightarrow d'$ and $d'$ occurs before $d$ in $L$, by Lemma~\ref{unforced}, we have that $d\miss d'$ is not an unforced missing edge. Now, if $d\miss d'$ is forced in both directions, by Lemma~\ref{duallyforced}, we have that $\{d,d'\}$ has an in-neighbor in $\Delta_D$, which is a contradiction. Therefore, we can conclude that $d\miss d'$ is singly-forced in $D$.	
As $d\dashrightarrow d'$ in $T$ and $T$ is a safe completion, it should be the case that $d\miss d'$ is singly-forced in the direction $d$ to $d'$ in $D$. Then by Definition~\ref{singly_forced}, there exists $v\in V(D)$ such that, $d\twoheadrightarrow v\rightarrow d'$. This together with the assumption that $d$ is prime implies that $dFv$. This shows that $k\geq 2$ (recall that $y_1=d$), which implies that $y_{k-1}$ exists.

As $y_{k-1}$ is the feed vertex of $L_{k-1}$ and $y_{k-1}Fy_k$, by Lemma~\ref{relation_miss}\ref{y'exists}, we have that there exists $y_k'\in V(T)$ such that $y_k\dashrightarrow y_k'\rightarrow y_{k-1}$ in $T$.
Since $y_k$ is the feed vertex of $L_k$, $y_k$ is prime, and $y_k\dashrightarrow y_k'$, by Lemma~\ref{noin-neighbor}, we have that $\{y_k,y_k'\}$ has no in-neighbor in $\Delta_D$. Now, suppose that $\{y_k,y_k'\}$ has an out-neighbor $\{b,b'\}$ in $\Delta_D$. Then, without loss of generality, we can assume that $(y_k,y_k')R(b,b')$, i.e. there exists the four cycle $y_k\rightarrow b\twoheadrightarrow y_k'\rightarrow b'\twoheadrightarrow y_k$. As $y_k$ is prime, $\{y_k,y'_k\}$ does not lie on any cycle in $\mathcal{C}$. Therefore, since $T$ is a safe completion, $y_k\dashrightarrow y_k'$, and $(y_k,y_k')R(b,b')$, we have that $b\dashrightarrow b'$. Since $y_{k-1}Fy_k$, there exists a vertex $y'_{k-1}$ such that $y_{k-1}\miss y'_{k-1}$ is singly-forced in the direction from $y_{k-1}$ to $y'_{k-1}$. As $T$ is a safe completion, this means that $y_{k-1}\dashrightarrow y'_{k-1}$. As the missing edge incident on $b'$ is oriented towards $b'$ in $T$, this implies that $b'\neq y_{k-1}$. Clearly, $b\neq y_{k-1}$ (as $y_{k-1}\twoheadrightarrow y_k$, but $y_k\rightarrow b$). Recalling that $b\miss b'$, we now have that either $b\rightarrow y_{k-1}$ or $y_{k-1}\rightarrow b$. Now if $b\rightarrow y_{k-1}$, then $b\rightarrow y_{k-1} \twoheadrightarrow y_k\rightarrow b$ would form a directed triangle containing a special arc and if $y_{k-1}\rightarrow b$, then $y_{k-1}\rightarrow b\twoheadrightarrow y_k'\rightarrow y_{k-1}$ would form a directed triangle containing a special arc. As we have a contradiction in both cases, $\{y_k,y'_k\}$ has no out-neighbor in $\Delta_D$. Therefore, $\{y_k,y_k'\}$ is an isolated vertex in $\Delta_D$.
	
By Lemma~\ref{unforced} applied on the median order $L_k$ of $T$, we have that $y_k\miss y_k'$ is not an unforced missing edge. As  $\{y_k,y_k'\}$ has no in-neighbor in $\Delta_D$, by Lemma~\ref{duallyforced}, $y_k\miss y_k'$ is not forced in both directions. Therefore, we can conclude that, $y_k\miss y_k'$ is singly-forced. As $y_k\dashrightarrow y'_k$ and $T$ is a safe completion, we know that $y_k\miss y'_k$ is forced in the direction $y_k$ to $y_k'$, i.e. there exists a vertex $u$ such that $y_k\twoheadrightarrow u\rightarrow y_k'$. As $y_k$ is prime, this further implies that $y_k F u$.
Then $(y_1,y_2,\ldots,y_k,y_{k+1}=u)$ is a sequence longer than $(y_1,y_2,\ldots,y_k)$ with the property that $y_1=d$ and $y_1Fy_2F\cdots Fy_{k+1}$. This contradicts our choice of the sequence $(y_1,y_2,\ldots,y_k)$.

\end{proof} 
\begin{lemma}\label{stable2}
Let $T$ be a maximum value safe completion of $D$ and let $L$ be a median order of $T$ having feed vertex $d$. If $|N^{++}_T(d)\setminus I(d)|> |N^+_T(d)\setminus I(d)|$, then $d$ has a large second neighborhood in $H$.
\end{lemma}
\begin{proof}
	If $z\notin N^+_H(d)$, then we are done by Lemma~\ref{twovertices}\ref{xnoutneigh}. So we can assume that $z\in N^+_H(d)$.
	
	If $d$ has no special in-neighbors or has a special in-neighbor of Type-I, then we are done by Lemma~\ref{stable1}. Therefore, we shall assume that $d$ has no Type-I special in-neighbors but has at least one Type-II special in-neighbor.
	Let $x$ be any Type-II special in-neighbor of $d$. Then there exists $d'\in V(T)$ such that $d\dashrightarrow d'\rightarrow x\twoheadrightarrow d$ in $T$. As $d$ has no Type-I special in-neighbors, by Lemma~\ref{dingamma}, we get that $d$ is not prime, i.e.  $I(d)=\Gamma(Q)$ for some $Q\in \mathcal{C}$. Therefore, by Corollary~\ref{gamma}\ref{missingingamma}, $d'\in \Gamma(Q)$. Now suppose that $x\notin \Gamma(Q)$. Then since $d'\rightarrow x$, $x\rightarrow d$ and $d,d'\in \Gamma(Q)$, we have a contradiction to the fact that $\Gamma(Q)$ is a module in $D$ (by Corollary~\ref{gamma}\ref{gammamodule}). Therefore, every special in-neighbor of $d$ is contained in $\Gamma(Q)=I(d)$; in other words, there are no special in-neighbors of $d$ in $N^{++}_T(d)\setminus I(d)$. Then by Lemma~\ref{good_vetices}, we have that $N^{++}_T(d)\setminus I(d) \subseteq N^{++}_D(d)\setminus I(d)$. By Remark~\ref{sncmodule}, we have $|N^+_D(d)\cap I(d)|= |N^{++}_D(d)\cap I(d)|$.
	Combining our observations, we get
	\begin{eqnarray*}
		|N^+_H(d)| &=& |N^+_D(d)|+1\mbox{ (as $z\in N^+_H(d)$)}\\
		&=& |N^+_D(d)\setminus I(d)|+ |N^+_D(d)\cap I(d)|+1 \\
		&\leq& |N^+_T(d)\setminus I(d)|+ |N^+_D(d)\cap I(d)|+1\quad (\text{since }N^+_D(d)\subseteq N^+_T(d))\\
		&\leq& |N^{++}_T(d)\setminus I(d)| + |N^{++}_D(d)\cap I(d)| \quad (\text{as } |N^{++}_T(d)\setminus I(d)|> |N^+_T(d)\setminus I(d)|)\\
		&\leq& |N^{++}_D(d)\setminus I(d)|+ |N^{++}_D(d)\cap I(d)| \quad (\text{as } N^{++}_T(d)\setminus I(d) \subseteq N^{++}_D(d)\setminus I(d))\\
		&=& |N^{++}_D(d)|\\
		&\leq& |N^{++}_H(d)|
	\end{eqnarray*}
	Therefore, $d$ has a large second neighborhood in $H$.
\end{proof}
\medskip

We are now ready to give a formal proof of Theorem~\ref{finalthm}.

\subsubsection*{Proof of Theorem~\ref{finalthm}}
\begin{proof}
Let $T$ be a maximum value safe completion of $D$.
By Lemma~\ref{module}, there exists a good median order $L$ of $T$ with respect to $\mathcal{I}(D)$.
If $L$ is periodic, then we are done by Lemma~\ref{periodic}. Therefore, we can assume that $L$ is stable.
Then, by the definition of a stable median order, there exists an integer $q\geq 0$ such that $Sed^{q+1}_{\mathcal{I}(D)}(L) = Sed^q_{\mathcal{I}(D)}(L)$. By Theorem~\ref{sedimentation}, $L'=Sed^q_{\mathcal{I}(D)}(L)$ is a median order of $T$. Let $d$ be the feed vertex of $L'$. By Theorem~\ref{mainthm}, $d$ has a large second neighborhood in $D$. As $Sed_{\mathcal{I}(D)}(L')=L'$, we have $|N^{++}_T(d)\setminus I(d)|> |N^+_{T}(d)\setminus I(d)|$. We can then conclude by Lemma~\ref{stable2} that $d$ has a large second neighborhood in $H$ as well.
\end{proof}
\begin{corollary}\label{cornearmatching}
Every oriented graph whose missing edges can be partitioned into a matching and a star contains a vertex with a large second neighborhood.
\end{corollary}

\begin{corollary}\label{cortwogood}
	Every oriented graph whose missing edges form a matching and does not contain a sink contains at least two vertices with large second neighborhoods.
\end{corollary}
\begin{proof}
	Let $H$ be an oriented graph whose missing edges form a matching and does not contain a sink. By Theorem~\ref{mainthm}, we know that there exists a vertex $z$ in $H$ with a large second neighborhood. As $H-\{z\}$ is an oriented graph whose missing edges form a matching, by Theorem~\ref{finalthm}, we can infer that there exists a vertex $z'\in V(H)\setminus\{z\}$ that has a large second neighborhood in $H$.
\end{proof}

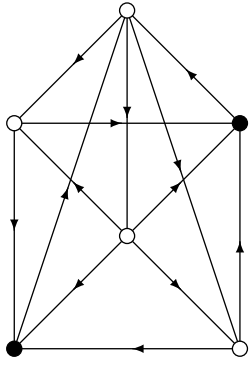
\begin{figure}[h]
\newcommand{\myptr}{{\arrow{stealth}}}
\renewcommand{\vertexset}{(a,0,0,black),(b,3,0),(c,3,3,black),(d,0,3),(x,1.5,1.5),(z,1.5,4.5)}
\renewcommand{\edgeset}{(b,a),(d,a),(x,a),(a,z),(b,c),(x,b),(z,b),(d,c),(x,c),(c,z),(x,d),(z,d),(z,x)}
\renewcommand{\defradius}{.1}
\renewcommand{\isdirected}{myptr}
\renewcommand{\arrowplacement}{0.475}
\begin{center}
\begin{tikzpicture}
\drawgraph
\end{tikzpicture}
\end{center}
\caption{A tournament missing a matching with no sink and exactly two vertices with large second neighborhoods (shown in black).}\label{fig}
\end{figure}

The graph shown in Figure~\ref{fig} is a tournament missing a matching without a sink that contains exactly two vertices with large second neighborhoods. Therefore, Corollary~\ref{cortwogood} is tight.

\section{Conclusion}

The question of whether there exists two vertices with large second neighborhoods in any oriented graph without a sink seems to be open.

\begin{conjecture}\label{suresh1}
Any oriented graph without a sink contains at least two vertices with large second neighborhoods.
\end{conjecture}

Clearly, Conjecture~\ref{suresh1} implies Conjecture~\ref{ssnc} (the Second Neighborhood Conjecture). We propose the following conjecture, which though apparently weaker at first sight, can be shown to be equivalent to Conjecture~\ref{suresh1}.

\begin{conjecture}\label{suresh}
If an oriented graph contains exactly one vertex with a large second neighborhood, then that vertex is a sink.
\end{conjecture}

It is easy to see that Conjecture~\ref{suresh1} implies Conjecture~\ref{suresh}.

\begin{proposition}
Conjecture~\ref{suresh} implies Conjecture~\ref{ssnc}.
\end{proposition}
\begin{proof}
Suppose that Conjecture~\ref{suresh} is true but Conjecture~\ref{ssnc} is not. Let $D$ be a minimal counterexample to Conjecture~\ref{ssnc}: i.e., $D$ is an oriented graph with minimum number of vertices and edges in which no vertex has a large second neighborhood. In particular, $D$ cannot have a sink. Let $(u,v)\in E(D)$. Consider the graph $D'$ obtained by removing the edge $(u,v)$ from $D$. As $D$ is a minimal counterexample to Conjecture~\ref{ssnc}, we know that $D'$ contains at least one vertex with a large second neighborhood. We claim that $D'$ contains at least two vertices with large second neighborhoods. Suppose not. Then by Conjecture~\ref{suresh}, $D'$ has a sink in it. As there is no sink in $D$ and every vertex other than $u$ has the same out-neighborhood in both $D$ and $D'$, this means that $u$ must be a sink in $D'$. Then, $N^+_D(u)=\{v\}$ and $N^+_D(v)\subseteq N^{++}_D(u)$. Since $v$ is not a sink in $D$, we also have that $|N^+_D(v)|\geq 1$, which gives us $|N^{++}_D(u)|\geq 1$. Therefore, $u$ has a large second neighborhood in $D$, which is a contradiction. This proves that $D'$ contains at least two vertices with large second neighborhoods. Then there exists a vertex $w\neq u$ in $D'$ such that $|N^+_{D'}(w)|\leq |N^{++}_{D'}(w)|$. As $w\neq u$, by the definition of $D'$, we have that $N^+_D(w)=N^+_{D'}(w)$ and $N^{++}_{D'}(w)\subseteq N^{++}_D(w)$. Combining this with the previous observation, we get $|N^+_D(w)|\leq |N^{++}_D(w)|$, implying that $w$ has a large second neighborhood in $D$, which is a contradiction. 
\end{proof} 

By the above proposition, if Conjecture~\ref{suresh} is true, then Conjecture~\ref{ssnc} is true, and these together imply that Conjecture~\ref{suresh1} is true. Thus Conjectures~\ref{suresh1} and~\ref{suresh} are equivalent. We do not know if these conjectures are equivalent to Conjecture~\ref{ssnc} or if they hold for the class of graphs studied in Section~\ref{2-deg}.
\section{Acknowledgements}
Part of this work was done when the first author was a postdoc at the Institute of Mathematical Sciences, Chennai. The first author would also like to thank the support provided by the NPIU TEQIP-III grant GO/TEQIP-III/EO/17-18/57.
\bibliographystyle{abbrv}

\newpage
\appendix
\section{Omitted proofs}\label{sec:omitted}

We first prove a proposition that will be used in the proof of Lemma~\ref{module} and Theorem~\ref{sedimentation}.
\begin{proposition}\label{compaction}
Let $(x_1,x_2,\ldots,x_n)$ be a median order of a tournament $T$. Let $i,j\in\{1,2,\ldots,n\}$ such that $i<j-1$ and $x_i$ and $x_j$ belong to a module in $T$ and every vertex in $\{x_{i+1},\ldots,x_{j-1}\}$ is outside this module. Then:
\begin{myenumerate}
\item\label{rightcompact} $(x_1,x_2,\ldots,x_{i-1},x_{i+1},x_{i+2},\ldots,x_{j-1},x_i,x_j,x_{j+1},\ldots,x_n)$ is a median order of $T$, and
\item\label{leftcompact} $(x_1,x_2,\ldots,x_i,x_j,x_{i+1},x_{i+2},\ldots,x_{j-1},x_{j+1},\ldots,$ $x_n)$ is a median order of $T$.
\end{myenumerate} 
\end{proposition}
\begin{proof}
Consider the set of vertices $X=\{x_{i+1},x_{i+2},\ldots,x_{j-1}\}$. Suppose that $|N^+(x_i)\cap X|>|N^-(x_i)\cap X|$. As $x_i$ and $x_j$ belong to a module in $T$ and every vertex of $X$ is outside this module, we have $N^+(x_j)\cap X=N^+(x_i)\cap X$ and $N^-(x_j)\cap X=N^-(x_i)\cap X$. This gives us $|N^+(x_j)\cap X|>|N^-(x_j)\cap X|$, which contradicts Proposition~\ref{feedback}\ref{xjf} applied on $x_{i+1}$ and $x_j$. Therefore, $|N^+(x_i)\cap X|\leq |N^-(x_i)\cap X|$. Then by Proposition~\ref{feedback}\ref{xif} applied on $x_i$ and $x_{j-1}$, we have $|N^+(x_i)\cap X|=|N^-(x_i)\cap X|$. Applying Proposition~\ref{movexi}\ref{xi} on $x_i$ and $x_{j-1}$, we now get that $(x_1,x_2,\ldots,x_{i-1},x_{i+1},x_{i+2},\ldots,x_{j-1},x_i,x_j,x_{j+1},\ldots,$ $x_n)$ is a median order of $T$. This proves~\ref{rightcompact}. It is easy to see, by repeating the same arguments for $x_j$ and $X$, that \ref{leftcompact} is also true.
\end{proof}

We are now ready to prove Lemma~\ref{module}, which is a slight variation of an observation of Ghazal~\cite{ghazal2015remark}.
\setcounter{lemma}{2}
\begin{lemma}
Let $\mathcal{I} = \{I_1, I_2,\ldots, I_r\}$ be a partition of the vertex set of a tournament $T$ into modules and let $L$ be a median order of $T$. Then there is a good median order $L'$ of $T$ with respect to $\mathcal{I}$ such that $L$ and $L'$ have the same feed vertex.
\end{lemma}
\begin{proof}
Given an ordering $P$ of the vertices of $T$ and a module $I\in\mathcal{I}$, a maximal subset of $I$ that is consecutive in $P$ is said to be a ``fragment'' of $I$ in $P$. Clearly, the fragments of a module $I\in\mathcal{I}$ are ordered from left to right in $P$. We define the ``weight'' of a vertex $v\in I$ with respect to $P$ to be the number of fragments of $I$ that occur after the fragment of $I$ containing $v$. The weight of $P$ is defined to be the sum of the weights of all the vertices with respect to $P$. Note that the median orders of $T$ with zero weight are exactly the good median orders of $T$ with respect to $\mathcal{I}$. Now suppose that $P$ is a median order of $T$ with non-zero weight. Then there exists $I\in\mathcal{I}$ and $u,v\in I$ such that they are not consecutive in $P$ and no vertex between them in $P$ belongs to $I$. Let $P'$ be the median order of $T$ obtained from $P$ by applying Proposition~\ref{compaction}\ref{rightcompact} to $P$, $u$ and $v$. It can be verified that the weight of $P'$ is strictly less than the weight of $P$ and that $P$ and $P'$ have the same feed vertex. This means that by applying the above procedure repeatedly to the median order $L$ of $T$, we can obtain a median order $L'$ of $T$ with zero weight (hence, it is a good median order of $T$ with respect to $\mathcal{I}$) having the same feed vertex as $L$.
\end{proof}

We shall now prove the following proposition and theorem which are adapted from the proof of Havet and Thomasse so as to incorporate our slightly changed definition of sedimentation.
\begin{proposition}\label{sed}
Let $T$ be a tournament and $L=(x_1,x_2,\ldots,x_n)$ be a median order of $T$ such that $|N^+(x_n)|=|N^{++}(x_n)|$.
\begin{myenumerate}
\item\label{movefeed} If $N^-(x_n)=N^{++}(x_n)$, then $(x_n,x_1,x_2,\ldots,x_{n-1})$ is a median order of $T$, and
\item\label{movebad} If $N^-(x_n)\setminus N^{++}(x_n)\neq\emptyset$ and $x_i$ is the vertex in $N^-(x_n)\setminus N^{++}(x_n)$ that occurs first in $L$, then $(x_i,x_1,x_2,\ldots,x_{i-1},x_{i+1},\ldots,x_n)$ is a median order of $T$.
\end{myenumerate}
\end{proposition}
\begin{proof}
\ref{movefeed} Since $|N^+(x_n)|=|N^{++}(x_n)|$ and $N^-(x_n)=N^{++}(x_n)$, we have $|N^+(x_n)|=|N^-(x_n)|$. Therefore, by Proposition~\ref{movexi}\ref{xj} applied on $x_1$ and $x_n$, we have that $(x_n,x_1,x_2,\ldots,x_{n-1})$ is a median order of $T$.

\ref{movebad} Let $D=\{x_1,x_2,\ldots,x_{i-1}\}$ and $U=\{x_i,x_{i+1},\ldots,x_n\}$.
By Proposition~\ref{subtournament}\ref{suborder}, $(x_i,x_{i+1},\ldots,x_n)$ is a median order of the subtournament $T[U]$ of $T$. Applying Theorem~\ref{thmht} to this median order of the tournament $T[U]$, we have $|N^+_T(x_n)\cap U|=|N^+_{T[U]}(x_n)|\leq|N^{++}_{T[U]}(x_n)|\leq |N^{++}_T(x_n)\cap U|$. This together with the fact that, $|N^+_T(x_n)|= |N^+_T(x_n)\cap D|+|N^+_T(x_n)\cap U|$, $|N^{++}_T(x_n)| = |N^{++}_T(x_n)\cap D|+|N^{++}_T(x_n)\cap U|$ and $|N^+_T(x_n)|=|N^{++}_T(x_n)|$ (assumption of the lemma) implies that $|N^+_T(x_n)\cap D|\geq |N^{++}_T(x_n)\cap D|$. As $x_i\in N^-_T(x_n)\setminus N^{++}_T(x_n)$, we have that $N^+_T(x_n)\cap D\subseteq N^+_T(x_i)\cap D$ and $N^-_T(x_i)\cap D\subseteq N^-_T(x_n)\cap D$. As $x_i$ is the first vertex in $L$ that belongs to $N^-_T(x_n)\setminus N^{++}_T(x_n)$, we also have that $N^-_T(x_n)\cap D=N^{++}_T(x_n)\cap D$. By Proposition~\ref{feedback}\ref{xif} applied to $x_1$ and $x_i$, we get $|N^+_T(x_i)\cap D|\leq |N^-_T(x_i)\cap D|$. Combining everything, we have $|N^+_T(x_n)\cap D|\leq |N^+_T(x_i)\cap D|\leq |N^-_T(x_i)\cap D|\leq |N^-_T(x_n)\cap D|=|N^{++}_T(x_n)\cap D|$. Recalling our previous observation that $|N^+_T(x_n)\cap D|\geq |N^{++}_T(x_n)\cap D|$, we then have $|N^+_T(x_i)\cap D|=|N^-_T(x_i)\cap D|$. Now from Proposition~\ref{movexi}\ref{xj} applied on $x_1$ and $x_i$, we get that $(x_i,x_1,x_2,\ldots,x_{i-1},x_{i+1},\ldots,x_n)$ is a median order of $T$.
\end{proof}

Following is the theorem from~\cite{havet2000median} that we need.

\begin{theorem}[Havet-Thomasse~\cite{havet2000median}]\label{thmhtsed}
Let $T$ be a tournament and $L=(x_1,x_2,\ldots,x_n)$ be a median order of it such that $|N^+(x_n)|=|N^{++}(x_n)|$. Let $b_1,b_2,\ldots,b_k$ be the vertices in $N^-(x_n)\setminus N^{++}(x_n)$ and $v_1,v_2,\ldots,v_{n-k-1}$ be the vertices in  $N^+(x_n)\cup N^{++}(x_n)$, both enumerated in the order in which they appear in $L$. Then $(b_1,b_2,\ldots,b_k,x_n,$ $v_1,v_2,\ldots,v_{n-k-1})$ is a median order of $T$.
\end{theorem}
\begin{proof}
We prove this by induction on $|N^-(x_n)\setminus N^{++}(x_n)|$.
If $|N^-(x_n)\setminus N^{++}(x_n)|=0$, then we are done by Proposition~\ref{sed}\ref{movefeed}. So let us assume that $N^-(x_n)\setminus N^{++}(x_n)\neq\emptyset$ and that $b_1,b_2,\ldots,b_k,v_1,v_2,\ldots,v_{n-k-1}$ are the vertices as defined in the statement of the theorem. By Proposition~\ref{sed}\ref{movebad}, we know that $\hat{L}=(x_i,x_1,\ldots,x_{i-1},x_{i+1},$ $\ldots,x_n)$ is a median order of $T$, where $x_i=b_1$. By Proposition~\ref{subtournament}\ref{suborder}, we know that $L'=(x_1,x_2,\ldots,x_{i-1},x_{i+1},$ $\ldots,x_n)$ is a median order of $T'=T-\{b_1\}$. It is easy to see that $N^+_{T'}(x_n)=N^+_T(x_n)$, $N^-_{T'}(x_n)=N^-_T(x_n)\setminus\{b_1\}$ and $N^{++}_{T'}(x_n)=N^{++}_T(x_n)$. Therefore, $|N^+_{T'}(x_n)|=|N^{++}_{T'}(x_n)|$ and $N^-_{T'}(x_n)\setminus N^{++}_{T'}(x_n)=\{b_2,b_3,\ldots,b_k\}$. By the induction hypothesis applied on the tournament $T'$ and the median order $L'$, we get that $(b_2,b_3,\ldots,b_k,x_n,v_1,v_2,$ $\ldots,v_{n-k-1})$ is a median order of $T'$. Now by Proposition~\ref{subtournament}\ref{replace}, we can replace the subsequence $(x_1,\ldots,x_n)$ of $\hat{L}$ with any median order of $T'$ to obtain a median order of $T$. Therefore, $(b_1,b_2,\ldots,b_k,x_n,v_1,v_2,\ldots,v_{n-k-1})$ is a median order of $T$.
\end{proof}

We are now ready to give a proof of Theorem~\ref{sedimentation}, which is again only a slight variation of a result of Ghazal~\cite{ghazal2015remark}.
\setcounter{theorem}{3}
\begin{theorem}
Let $T$ be a tournament. If $\mathcal{I}$ is a partition of $V(T)$ into modules and $L$ is a good median order of $T$ with respect to $\mathcal{I}$, then $Sed_{\mathcal{I}}(L)$ is also a good median order of $T$ with respect to $\mathcal{I}$.
\end{theorem}
\begin{proof}
Let $L=(x_1,x_2,\ldots,x_n)$ and let $I\in\mathcal{I}$ be the module containing $x_n$. If $|N^+_T(x_n)\setminus I|<|N^{++}_T(x_n)\setminus I|$, then $Sed_{\mathcal{I}}(L)=L$ and there is nothing to prove. Therefore, by Proposition~\ref{modht}, we can assume that $|N^+_T(x_n)\setminus I|=|N^{++}_T(x_n)\setminus I|$. Let $t=|I|$. Then $I=\{x_{n-t+1},x_{n-t+2},\ldots,x_n\}$. Let $b_1,b_2,\ldots,b_k$ be the vertices outside $I$ that are in-neighbors of $x_n$ but not its second out-neighbors (i.e., $\{b_1,b_2,\ldots,b_k\}=(N^-_T(x_n)\setminus N^{++}_T(x_n))\setminus I$), where $0\leq k\leq n-t$, and $v_1,v_2,\ldots,v_{n-t-k}$ the vertices in $(N^+_T(x_n)\cup N^{++}_T(x_n))\setminus I$, both enumerated in the order in which they appear in $L$.

For ease of notation, we denote $x_{n-t+i}$ by $u_i$, for each $i\in\{1,2,\ldots,t\}$. Then $u_1=x_{n-t+1}$ and $u_t=x_n$. By Proposition~\ref{subtournament}\ref{suborder}, $L'=(x_1,x_2,\ldots,x_{n-t+1}=u_1)$ is a median order of $T'=T-\{u_2,u_3,\ldots,u_t\}$. As $u_1$ and $x_n$ belong to the module $I$ of $T$, $N^+_{T'}(u_1)=N^+_T(u_1)\setminus I=N^+_T(x_n)\setminus I$ and $N^-_{T'}(u_1)=N^-_T(u_1)\setminus I=N^-_T(x_n)\setminus I$. By Proposition~\ref{modulesec}, we further have $N^{++}_{T'}(u_1)=N^{++}_T(u_1)\setminus I=N^{++}_T(x_n)\setminus I$. Since $|N^+_T(x_n)\setminus I|=|N^{++}_T(x_n)\setminus I|$, it then follows that $|N^+_{T'}(u_1)|=|N^{++}_{T'}(u_1)|$ and that $N^-_{T'}(u_1)\setminus N^{++}_{T'}(u_1)=\{b_1,b_2,\ldots,b_k\}$.

By Theorem~\ref{thmhtsed} applied on $T'$ and $L'$, we get that $(b_1,b_2,\ldots,b_k,u_1,v_1,v_2,\ldots,v_{n-t-k})$ is a median order of $T'$. From Proposition~\ref{subtournament}\ref{replace}, we know that we can replace the subsequence $(x_1,x_2,\ldots,x_{n-t+1}=u_1)$ of $L$ with this new median order of $T'$ to get another median order $(b_1,b_2,\ldots,b_k,u_1,v_1,v_2,\ldots,v_{n-t-k},u_2,u_3,\ldots,u_t)$ of $T$. By repeatedly applying Proposition~\ref{compaction}\ref{leftcompact} on the median order $(b_1,b_2,\ldots,b_k,u_1,u_2,\ldots,u_i,v_1,v_2,\ldots,$ $v_{n-t-k},u_{i+1},u_{i+2},\ldots,u_t)$ of $T$ and the vertices $u_i$ and $u_{i+1}$, for each value of $i$ from 1 to $t-1$, we can conclude that $Sed_{\mathcal{I}}(L)=(b_1,b_2,\ldots,b_k,u_1,u_2,\ldots,u_t,v_1,v_2,\ldots,v_{n-t-k})$ is a median order of $T$.

It only remains to be proven that $Sed_{\mathcal{I}}(L)$ is a good median order of $T$ with respect to $\mathcal{I}$. It can be easily seen that for any $J\in\mathcal{I}$, if there exists $u\in J$ such that $u\in N^-_T(x_n)\setminus N^{++}_T(x_n)=\{b_1,b_2,\ldots,b_k\}$, then $J\subseteq N^-_T(x_n)\setminus N^{++}_T(x_n)$. As $\{u_1,u_2,\ldots,u_t\}=I\in\mathcal{I}$, this implies that every other module in $\mathcal{I}$ is a subset of either $\{b_1,b_2,\ldots,b_k\}$ or $\{v_1,v_2,\ldots,v_{n-t-k}\}$. Since the vertices in each set in $\mathcal{I}$ occur in $Sed_{\mathcal{I}}(L)$ in the same order as they occur in $L$, and $L$ is a good median order of $T$ with respect to $\mathcal{I}$, we can conclude that the vertices in each module in $\mathcal{I}$ occur consecutively in $Sed_{\mathcal{I}}(L)$ too.
\end{proof}

\end{document}